\documentclass[journal,twoside,web]{ieeecolor}

\usepackage{generic}
\usepackage{cite}
\usepackage{amsmath,amssymb,amsfonts}
\usepackage{algorithmic}
\usepackage{graphicx}
\usepackage{algorithm,algorithmic}
\usepackage{hyperref}
\usepackage{textcomp}
\usepackage{nicefrac}
\usepackage{mathtools}

\usepackage{fancyhdr}
\usepackage{xcolor}

\hypersetup{
	pdftitle={Muntwiler2025 - MHE under parametric uncertainty --- Robust state estimation without informative data},
	pdfauthor={Simon Muntwiler, Johannes Köhler, Melanie N. Zeilinger},
	colorlinks,
	linkcolor={red!50!black},
	citecolor={blue!50!black},
	urlcolor={blue!80!black}
}

\def\doi{10.1109/TAC.2025.3633537}
\fancyfootoffset{-1em}
\fancypagestyle{copyright}{
	\vspace{-1.8em}\cfoot{\footnotesize\textcopyright{} 2025 IEEE. Personal use of this material is permitted.  Permission from IEEE must be obtained for all other uses, in any current or future media, including reprinting/republishing this material for advertising or promotional purposes, creating new collective works, for resale or redistribution to servers or lists, or reuse of any copyrighted component of this work in other works.}
	\chead{\footnotesize This article has been accepted for publication in IEEE Transactions on Automatic Control. This is the author's version which has not been fully edited and content may change prior to final publication. Citation information: \href{http://dx.doi.org/\doi}{DOI \doi}}
}

\newtheorem{assumption}{Assumption}
\newtheorem{definition}{Definition}
\newtheorem{theorem}{Theorem}
\newtheorem{proposition}{Proposition}

\newtheorem{corollary}{Corollary}
\newtheorem{remark}{Remark}

\def\BibTeX{{\rm B\kern-.05em{\sc i\kern-.025em b}\kern-.08em
    T\kern-.1667em\lower.7ex\hbox{E}\kern-.125emX}}
\begin{document}
\title{MHE under parametric uncertainty --- Robust state estimation without informative data}
\author{Simon Muntwiler, Johannes Köhler, Melanie N. Zeilinger
\thanks{The work of Simon Muntwiler was supported by the Bosch Research Foundation im Stifterverband. 
	This work has been supported by the Swiss National Science Foundation under NCCR Automation (grant agreement 51NF40\_180545).}
\thanks{The authors are with the Institute for Dynamic Systems and Control, ETH Zürich, Switzerland. (email: $\{\text{simonmu,jkoehle,mzeilinger}\}$@ethz.ch)}}

\maketitle
\thispagestyle{copyright}

\begin{abstract}
In this paper, we study joint state and parameter estimation for general nonlinear systems with uncertain parameters and persistent process and measurement noise. 
In particular, we are interested in stability properties of the resulting state estimate in the absence of persistency of excitation (PE).
With a simple academic example, we show that existing moving horizon estimation (MHE) approaches for joint state and parameter estimation as well as classical adaptive observers can result in diverging state estimates in the absence of PE, even if the noise is small.
We propose an MHE formulation involving a regularization based on a constant prior estimate of the unknown system parameters.
Only assuming the existence of a stable state estimator, we prove that the proposed MHE approach results in
practically robustly stable state estimates irrespective of PE.
We discuss the relation of the proposed MHE formulation to state-of-the-art results from MHE and adaptive estimation.
The properties of the proposed MHE approach are illustrated with a numerical example of a car with unknown tire friction parameters.
\end{abstract}


\section{Introduction}
State estimation is of high importance for control, system monitoring, and fault detection in many applications such as power systems~\cite{abur2004power}, battery management systems~\cite{Wang2020a}, manufacturing systems~\cite{Pasadyn2005}, or biomedical systems~\cite{Kashif2012}.
Obtaining an accurate estimate of the state in general requires access to a model which accurately describes the underlying system dynamics.
Identifying the parameters of such a model exactly is, however, difficult, e.g., due to the nonlinear nature of the models.
In addition, the parameters are often slowly varying, e.g., due to aging of components, and therefore practical applications require some form of online adaptation of the model parameters.
This renders simultaneous estimation of states and parameters of nonlinear systems from noisy output measurements a problem of high practical relevance, compare, e.g,~\cite{ Zanon2013,Boegli2013,Poloni2010,Abdollahpouri2017,Russo1999,Kuhl2011,Kupper2009,Valluru2017,Chen2012,Tuveri2023,kleyman2023state,Frick2012}.
However, the resulting estimation problem is nonlinear and non-convex even in the simple case of linear systems, see, e.g.,~\cite{Ljung1979,Gibson2005}, and thus challenging to solve.
Furthermore, unique identification of model parameters requires sufficiently informative data, typically through a so-called persistency of excitation (PE) condition (cf.~\cite{Anderson1977,Narendra1987,Willems2005,Cho1997}), even if noise-free state measurements are available.
Reliable parameter estimation is not possible in the absence of PE or, more generally, if the excitation is small compared to the noise, i.e., the signal to noise ratio is small.
For general nonlinear systems or over-parametrized models, ensuring PE is typically impossible. 
For linear systems, PE can be ensured using additional excitation~\cite{Willems2005}, however, this requires active deterioration of the regulation performance.
In this paper, we present an estimation method with guaranteed robust stability for general nonlinear systems with unknown parameters subject to process and measurement noise, without relying on a PE assumption. 
\subsubsection*{Related work}
Adaptive observers are classical methods for online state estimation under parametric uncertainty~\cite{Kudva1973,Kudva1974,Kreisselmeier1977,Shahrokhi1982,Ticlea2016,Bastin1988,Marino1995,Zhang2002,Marino2001,Tomei2022,Dey2022,Wang2020,Ortega2020,Katiyar2022}.
Assuming PE, convergence of state and parameter estimates resulting from such an adaptive observer can be shown for linear systems~\cite{Kudva1973,Kudva1974,Kreisselmeier1977,Shahrokhi1982,Ticlea2016} and selected classes of nonlinear systems such as SISO systems in observable canonical form~\cite{Bastin1988} or linearly parameterized SISO systems~\cite{Marino1995}.
Additionally, robust stability under process and measurements noise was established for adaptive observers for linear systems in~\cite{Marino2001,Zhang2002}.
Under a less restrictive initial excitation (IE) condition, e.g., dynamic regressor extension and mixing~\cite{Wang2020,Ortega2020} or concurrent learning~\cite{Ortega2020} lead to stable estimates in linear systems without noise.
Assuming bounded noise and IE, robustness of the estimation under process and measurement noise was analyzed for linear systems in~\cite{Katiyar2022} by employing a switching mechanism in the parameter estimator.
In the absence of PE, boundedness of state and parameter estimation errors in the noise-free case, i.e., without process and measurement noise, were shown for linear systems~\cite{Dey2022,Ortega2020} and a special class of nonlinear systems~\cite{Marino2001,Tomei2022}, where the nonlinearity in the systems dynamics depends only on inputs and measurements.
Ensuring robustness of state estimates under parametric uncertainty in the absence of PE is in general an open problem, which is challenging to address since PE is virtually necessary for robustness of parameter estimates~\cite{Johnstone1982}.
Overall, existing adaptive observer designs are mostly applicable to linear systems or restrictive classes of nonlinear systems, and typically lack robustness guarantees in the absence of PE.

More recently, joint state and parameter estimation for nonlinear systems has been addressed using optimization-based approaches, see, e.g.,~\cite{Simpson2023,Brouillon2022} for the case of linear systems.
Moving horizon estimation (MHE) is a promising optimization-based technique for estimation in general nonlinear systems subject to process and measurement noise~\cite{Rawlings2020,Schiller2023,Allan2020a}.
In case of perfectly known model parameters, uniform\footnote{Uniform in the control input.} robust stability of MHE for state estimation can be established with a sufficiently large horizon under a general detectability condition in the form of incremental input/output-to-state-stability ($\delta$-IOSS)~\cite{Allan2021a,Schiller2023}.
In this setting, it has been shown recently that $\delta$-IOSS is a necessary and sufficient condition for the existence of a robustly stable state estimator~\cite{Allan2021}.
In addition, MHE has been shown to perform well even in situations where standard approaches like an extended Kalman filter (EKF) lead to non-physical state estimates~\cite{haseltine2005critical}.
Compared to classical estimation approaches, MHE relies on increased computational power to solve a finite-horizon optimization problem online.
Consequently, this lead to the development of efficient solution strategies~\cite{Kuhl2011,Kouzoupis2016,Baumgartner2019,Muntwiler2022}.
Joint state and parameter estimation in MHE is commonly addressed by introducing an augmented state containing both the original system states and the unknown parameters~\cite{Robertson1996, Zanon2013,Boegli2013,Poloni2010,Abdollahpouri2017,Russo1999,Kuhl2011,Kupper2009,Valluru2017,Chen2012,Tuveri2023,kleyman2023state,Frick2012}.
Such MHE approaches are employed in different applications ranging from autonomous driving and friction estimation~\cite{Zanon2013,Boegli2013}, estimation in vibrating structures~\cite{Poloni2010,Abdollahpouri2017}, chemical processes~\cite{Russo1999,Kuhl2011,Kupper2009,Valluru2017}, biomedical systems~\cite{Chen2012,Tuveri2023,kleyman2023state}, to induction motors~\cite{Frick2012}.
Uniform detectability ($\delta$-IOSS) of the augmented state consisting of the system state and parameters in general requires a uniform PE condition~\cite{Sui2011,schiller2023nonlinear}, which can often not be ensured in practice~\cite{Zanon2013,Boegli2013,Poloni2010,Abdollahpouri2017,Russo1999,Kuhl2011,Kupper2009,Valluru2017,Chen2012,Tuveri2023,kleyman2023state,Frick2012}.
Consequently, the (robust) stability properties of these MHE schemes deteriorate arbitrarily if the excitation becomes small during online operation.
Additionally, applying an MHE for joint state and parameter estimation in the absence of PE can result in numerical issues.
While this has been addressed by modifications of the prior weighting in the MHE objective, compare~\cite{Sui2011,Baumgartner2022}, the resulting MHE schemes do not provide any stability guarantees in the absence of PE.
An MHE for joint estimation of states and parameters of a neural network system model, which generally suffers from lack of PE due to over-parametrization, involving a regularization on the parameter estimate was recently proposed in~\cite{lowenstein2023physics}.
However, theoretical properties of the resulting estimates were not established.
Using recent results from functional estimation~\cite{Muntwiler2023}, stability of the state estimate can be ensured even if the augmented state is not detectable.
However, this requires the use of a full information estimator (FIE), i.e., an optimization-based estimator with unbounded horizon.
The design of finite horizon MHE approaches with robust stability properties in the absence of PE for general nonlinear systems with parametric uncertainty and process and measurement noise remains an open question.
Note that an MHE approach for joint state and parameter estimation, which ensures bounded state estimation errors in the absence of PE, was proposed concurrently in~\cite{schiller2023nonlinear}.
While the proposed work uses a necessary and sufficient detectability condition, the \emph{robust} state detectability condition required in~\cite[Ass.~1]{schiller2023nonlinear} is generally more restrictive, compare also our discussion in Section~\ref{sec:detectability} below.
\subsubsection*{Contribution}
In this paper, we derive an MHE formulation for general nonlinear systems with parametric uncertainty subject to process and measurement noise which ensures reliable state estimation in the absence of PE, and consequently without unique identifiability of the parameters.
With a simple academic example, we show that standard MHE schemes for joint state and parameter estimation can lead to parameter drift in the absence of PE and thus lack fundamental robust stability properties (Section~\ref{sec:motivating_example}).
We also demonstrate that classical adaptive observers have the same limitation (Section~\ref{sec:discussion_adaptive}).
To overcome this, we propose to replace the standard prior weighting on the parameter estimate within the MHE formulation by a simple regularization involving a (constant) prior parameter estimate (Section~\ref{sec:estimation}).
We show practical robust stability of the resulting state estimation error under the proposed MHE approach.
Robust stability implies that the estimation error is proportional to the process and measurement noise. 
The stability is practical with respect to (w.r.t.) a constant involving the bias of the parameter regularization, which can be made arbitrarily small by increasing the horizon.
To establish robust stability, we only require a weak detectability condition on the system: There exists an estimator satisfying the desired robust stability property.
Hence, the proposed MHE approach is applicable whenever \emph{any} robustly stable estimator for the considered system exists.
Notably, satisfaction of this assumption does not require detectability of the augmented system state and hence does not impose a restrictive PE condition.
The presented theoretical analysis relies on recent techniques developed in the context of MHE~\cite{Schiller2023} and generalizations thereof to functional estimation~\cite{Muntwiler2023}.
We discuss the relation of the proposed MHE approach to classical MHE schemes and existing prior weightings to handle lack of PE, and connect the proposed ``regularization" to classical tools from adaptive control and estimation (Section~\ref{sec:discussion}).
Finally, we present a numerical example of a car with unknown tire friction parameters, with insufficient excitation over certain time periods (Section~\ref{sec:num}).
\subsubsection*{Notation}
Let the non-negative real numbers be denoted by $\mathbb{R}_{\geq 0}$, the set of integers by $\mathbb{I}$, the set of all integers greater than or equal to $a$ for some $a \in \mathbb{R}$ by $\mathbb{I}_{\geq a}$, and the set of integers in the interval $[a,b]$ for some $a,b\in\mathbb{R}$ with $a\le b$ by $\mathbb{I}_{[a,b]}$.
Let $\|x\|$ denote the Euclidean norm of the vector $x \in \mathbb{R}^n$.
A positive definite (semi-definite) matrix $P=P^\top$ is denoted as $P\succ 0$ ($P\succeq 0$).
The quadratic norm with respect to a positive definite matrix $Q=Q^\top$ is denoted by $\|x\|_Q^2=x^\top Q x$, and the minimal and maximal eigenvalues of $Q$ are denoted by $\lambda_{\min}(Q)$ and $\lambda_{\max}(Q)$, respectively.
The maximum generalized eigenvalue of positive definite matrices $A=A^\top$ and $B=B^\top$ is denoted as $\lambda_{\max}(A,B)$, i.e., the largest scalar $\lambda$ satisfying $\det (A-\lambda B) = 0$.
The identity matrix is denoted by $I_n\in\mathbb{R}^{n\times n}$. 
The $i$-th element of a vector $x\in\mathbb{R}^n$ is denoted as $[x]_i$ for $i \in \mathbb{I}_{[1,n]}$.

\section{Setup and Preliminaries} \label{sec:setup}
In this section, we formulate the considered state estimation problem, introduce the desired stability property, and state the necessary assumption to ensure the existence of an appropriate estimator.
\subsection{Problem Formulation} \label{sec:problem_formulation}
We consider a nonlinear discrete-time system
\begin{subequations}\label{eq:sys}
	\begin{align}
	x_{t+1}&=f(x_t,u_t,w_t,\theta), \label{eq:sys_1}\\
	y_t&=h(x_t,u_t,w_t,\theta), \label{eq:sys_2}
	\end{align}
\end{subequations}
where $t\in\mathbb{I}_{\ge 0}$ is the discrete time step, $x_t \in \mathbb{R}^{n_{\mathrm{x}}}$ the system state, $u_t\in\mathbb{U}=\mathbb{R}^{n_{\mathrm{u}}}$ the control input, $w_t\in\mathbb{R}^{n_{\mathrm{w}}}$ the process and measurement noise, $\theta \in \mathbb{R}^{n_\theta}$ a constant but unknown parameter, and $y_t\in\mathbb{Y}=\mathbb{R}^{n_{\mathrm{y}}}$ the measurement output.
Note that $w_t$ appears in the dynamics~\eqref{eq:sys_1} and measurement model~\eqref{eq:sys_2} and hence can also model separate process disturbances and measurement noise.

We consider the general case where we have additional system information in the form of constraints
\begin{align}\label{eq:constraints}
	x_t \in  \mathbb{X},\ w_t\in\mathbb{W},\ \forall t\in \mathbb{I}_{\ge 0},\ \theta \in \Theta,
\end{align}
with $\mathbb{X}\subseteq\mathbb{R}^{n_{\mathrm{x}}}$, $\mathbb{W}\subseteq\mathbb{R}^{n_{\mathrm{w}}}$, and $\Theta\subseteq\mathbb{R}^{n_\theta}$.
Note that this is not restrictive, since we can always choose $\mathbb{X}=\mathbb{R}^{n_{\mathrm{x}}}$, $\mathbb{W}=\mathbb{R}^{n_{\mathrm{w}}}$, and $\Theta = \mathbb{R}^{n_\theta}$ in case no additional system knowledge is available.
Constraints of the form~\eqref{eq:constraints} can be used to include additional physical prior knowledge, e.g., non-negativity of certain states, or bounds on noise and parameters, into a state estimator, compare~\cite[Sec.~4.3.3]{Rawlings2020}.
In the following, we assume that the true trajectory of system~\eqref{eq:sys} always satisfies the constraints~\eqref{eq:constraints}.

Given some available prior estimates $\bar{x}_0$ and $\bar{\theta}_0$ of the initial state $x_0$ and the unknown parameter $\theta$, respectively, the objective is to obtain an estimate $\hat{x}_t$ of the current system state $x_t$ at each time step $t\in\mathbb{I}_{\ge 0}$.
The considered problem is formalized by the following definition of a state estimator.
\begin{definition}[\protect{State estimator, adapted from~\cite[Def.~2.2]{Allan2021}}]\label{def:state_estimator}
	A state estimator is a sequence of functions $\Psi_t:\mathbb{X}\times\Theta\times\mathbb{W}^{t}\times\mathbb{Y}^{t}\rightarrow\mathbb{X}$ for $t\in\mathbb{I}_{\ge 0}$ to compute estimates of~\eqref{eq:sys_1} at each time step $t$ as
	\begin{align}\label{eq:state_estimator}
	\hat{x}_t = \Psi_t\left(\bar{x}_{0},\bar{\theta}_0,\{\bar{w}_j\}_{j=0}^{t-1},\{\bar{y}_j\}_{j=0}^{t-1}\right),
	\end{align}
	where $\bar{w}_j$ and $\bar{y}_j$ are estimates of the noise $w_j$ and measurements $y_j$ at time step $j$, respectively. 
\end{definition}

In a practical application, the noise estimates $\bar{w}_t$ are typically chosen to be zero and the output measurement $\bar{y}_j$ equal to the noisy output measurement $y_j$ according to~\eqref{eq:sys_2}.
The aim is to design a state estimator according to Definition~\ref{def:state_estimator} for system~\eqref{eq:sys} which satisfies the following stability notion.
\begin{definition}[Incremental input-to-output practical stability] \label{def:dIOpS}
	A state estimator according to Definition~\ref{def:state_estimator} is exponentially incrementally input-to-output practically stable ($\delta$-IOpS) with respect to some constant $\epsilon \ge 0$ if there exist $\lambda_1, \lambda_2, \lambda_3 \in [0,1)$ and $C_1, C_2, C_3 > 0$ such that
	\begin{align}
	\|x_t - \hat{x}_t\| \le& \max\left\{C_1\lambda_1^t\left\|\begin{bmatrix}
	x_0  \\
	\theta
	\end{bmatrix} - \begin{bmatrix}
	\bar{x}_0 \\
	\bar{\theta}_0
	\end{bmatrix}\right\|\right. ,\nonumber \\
	&\max_{j\in\mathbb{I}_{[0,t-1]}}C_2\lambda_2^{t-j-1}\|w_j-\bar{w}_j\|, \label{eq:ISpS_max_form} \\
	&\left. \max_{j\in\mathbb{I}_{[0,t-1]}}C_3\lambda_3^{t-j-1}\|y_j-\bar{y}_j\|, \epsilon \vphantom{\begin{bmatrix} x_0  \\ \theta\end{bmatrix}}\right\}. \nonumber
	\end{align}
	A state estimator according to Definition~\ref{def:state_estimator} is exponentially incrementally input-to-output stable ($\delta$-IOS) if additionally $\epsilon = 0$.
\end{definition}

This stability notion is an extension of the one defined in~\cite[Def.~3]{Muntwiler2023} in the context of functional estimation, compare also~\cite[Def.~6]{Limon2009} for classical input-to-state practical stability.
In contrast to classical stability, Inequality~\eqref{eq:ISpS_max_form} only provides a bound on the estimation error of the state $x_t$, but not the augmented state $\begin{bmatrix} x_t^\top & \theta^\top \end{bmatrix}^\top$.
Note that in case the true parameter is known, i.e., $\bar{\theta}_0 = \theta$, and~\eqref{eq:ISpS_max_form} is satisfied with $\epsilon = 0$, Definition~\ref{def:dIOpS} corresponds to a standard robust stability notion considered in MHE, compare~\cite[Def.~2]{Schiller2023}.
To design an $\delta$-IOpS state estimator, we require the following assumption toward the dynamics of the considered system. 
\begin{assumption}[Existence of stable state estimator]\label{ass:existence_of_estimator}
	There exists an exponentially $\delta$-IOS state estimator according to Definition~\ref{def:dIOpS}.
\end{assumption}

Note that only the existence of such an estimator is assumed, but it does not need to be known or have finite complexity.
In Section~\ref{sec:verify_ass} below, we establish the existence of such an estimator for a common class of systems.

\subsection{Detectability}\label{sec:detectability}

In the following, we introduce an incremental input/output-to-output-stability ($\delta$-IOOS) Lyapunov function, which characterizes existence of a stable estimator (Assumption~\ref{ass:existence_of_estimator}).

\begin{definition}[$\delta$-IOOS Lyapunov function\protect{~\cite[Def.~6]{Muntwiler2023}}]
	\label{def:dIOOS_Lyap}
	A function $W_\delta:\mathbb{R}^{n_{\mathrm{x}}}\times\mathbb{R}^{n_{\mathrm{x}}}\times\mathbb{R}^{n_\theta}\times\mathbb{R}^{n_\theta}\rightarrow\mathbb{R}_{\geq 0}$ is an exponential $\delta$-IOOS Lyapunov function if there exist $P_1,\ P_{2,\mathrm{x}},\ P_{2,\theta}\succ 0$, $Q,\ R\succeq 0$, and $\eta \in [0,1)$ such that
	\begin{subequations}
		\label{eq:dIOOS_Lyap}
		\begin{align}
		\label{eq:dIOOS_Lyap_1}
		&\|x-\tilde{x}\|_{P_1}^2 \le W_\delta(x,\tilde{x},\theta,\tilde{\theta}) \le \|x-\tilde{x}\|_{P_{2,\mathrm{x}}}^2 + \|\theta-\tilde{\theta}\|_{P_{2,\theta}}^2,\\ \nonumber
		&W_\delta\left(x^+, \tilde{x}^+, \theta, \tilde{\theta} \right) \leq \eta W_\delta(x,\tilde{x},\theta,\tilde{\theta}) \\
		&\hspace{4cm} +\|w-\tilde{w}\|_Q^2+\|y - \tilde{y}\|_{R}^2, \label{eq:dIOOS_Lyap_2}
		\end{align}
	\end{subequations}
	for all $\{x,u,w,y\}\ \in\ \mathbb{X}\times\mathbb{U}\times\mathbb{W}\times\mathbb{Y}$, $\{\tilde{x},u,\tilde{w},\tilde{y}\}\ \in\ \mathbb{X}\times\mathbb{U}\times\mathbb{W}\times\mathbb{Y}$, $\theta\in\Theta$, and $\tilde{\theta}\in\Theta$, where $x^+ = f(x,u,w,\theta)$, $\tilde{x}^+=f(\tilde{x},u,\tilde{w},\tilde{\theta})$, $y=h(x,u,w,\theta)$, and $\tilde{y}=h(\tilde{x},u,\tilde{w},\tilde{\theta})$.
\end{definition}

The following proposition shows that the existence of a $\delta$-IOOS Lyapunov function according to Definition~\ref{def:dIOOS_Lyap} for the system~\eqref{eq:sys} is a necessary and sufficient condition for Assumption~\ref{ass:existence_of_estimator} to be satisfied for the considered system.
\begin{proposition}\label{prop:existence_dIOOS_Lyap}
		Assumption~\ref{ass:existence_of_estimator} holds if and only if system~\eqref{eq:sys} admits an exponential $\delta$-IOOS Lyapunov function according to Definition~\ref{def:dIOOS_Lyap}.
\end{proposition}

\begin{proof}
		The considered problem is a special case of functional estimation studied in~\cite{Muntwiler2023} with system state $\begin{bmatrix} x_t^\top & \theta^\top \end{bmatrix}^\top$ and corresponding functional output $z_t = \phi(x_t,\theta) = x_t$ to be estimated.
		Resulting from~\cite[Cor.~8]{Muntwiler2023}, system~\eqref{eq:sys} admits an exponential $\delta$-IOOS Lyapunov function according to Definition~\ref{def:dIOOS_Lyap} if and only if the system is exponentially $\delta$-IOOS according to~\cite[Def.~4]{Muntwiler2023}.
		Additionally, it was shown in~\cite{Muntwiler2023} that exponential $\delta$-IOOS\footnote{Both \cite[Prop.~5]{Muntwiler2023} and \cite[Cor.~10]{Muntwiler2023} consider a more general asymptotic form of $\delta$-IOOS, but the proof naturally extends to the case of exponential $\delta$-IOOS.} is both necessary (cf.~\cite[Prop.~5]{Muntwiler2023}) and sufficient (cf.~\cite[Cor.~10]{Muntwiler2023}) for Assumption~\ref{ass:existence_of_estimator} to be satisfied.
		Consequently, Assumption~\ref{ass:existence_of_estimator} is necessary and sufficient for the existence of an exponential $\delta$-IOOS Lyapunov function.
\end{proof}

Note that we restrict our analysis to an \emph{exponential-like} version of $\delta$-IOOS here.
In~\cite[Thm.~9]{Muntwiler2023} a more general asymptotic version of $\delta$-IOOS was used to establish stability of a FIE approach for functional estimation.
Notably, also without parametric uncertainty, establishing stability for an MHE scheme with finite horizon generally requires a more restrictive (exponential) condition, compare~\cite[Sec.~5.5.3]{Allan2020a} and~\cite[Rem.~4]{Schiller2023}.

A related detectability condition was introduced concurrently in~\cite[Ass.~1]{schiller2023nonlinear}.
Since the detectability condition invoked by Assumption~\ref{ass:existence_of_estimator} is necessary and sufficient, the condition in~\cite{schiller2023nonlinear} \emph{cannot} be less restrictive. 
In particular, \cite[Ass.~1]{schiller2023nonlinear} considers a \emph{robust} state detectability, which uses a $\delta$-IOSS Lyapunov function for the state and holds uniformly and robustly w.r.t. the parameters $\theta$.
In contrast, Proposition~\ref{prop:existence_dIOOS_Lyap} shows that the existence of a stable estimator under parametric uncertainty is equivalent to a $\delta$-IOOS Lyapunov function of the augmented state $(x,\theta)$. 
In the trivial case of a single possible model parameter $\Theta=\{\theta\}$, these two definitions coincide and reduce to a classical state detectability condition~\cite{Allan2021}, but for the considered problem of uncertain model parameters they differ.

\subsection{Standard MHE for State and Parameter Estimation}\label{sec:standard_MHE}
In the following, we introduce a standard MHE formulation to jointly estimate states and parameters, which is widely applied, see, e.g.,~\cite{Robertson1996,Zanon2013,Boegli2013,Poloni2010,Abdollahpouri2017,Russo1999,Kuhl2011,Valluru2017,Kupper2009,Chen2012,Frick2012,kleyman2023state,Tuveri2023}.
The MHE approach at time step~$t$ considers past input data $\{u_j\}_{j=t-M_t}^{t-1}$ and estimates of past noises and outputs $\{\bar{w}_j,\bar{y}_j\}_{j=t-M_t}^{t-1}$ within a  window of length $M_t = \min\{t,M\}$, with horizon $M \in \mathbb{I}_{\ge 0}$, and the past state and parameter estimates\footnote{Typically, this choice is denoted as filtering prior, cf.~\cite[Chap.~4]{Rawlings2020}.}, $\hat{x}_{t-M_t}$ and $\hat{\theta}_{t-M_t}$, respectively.
Thereby, the MHE scheme optimizes over a sequence of $M_t$ disturbance estimates $\hat{w}_{\cdot | t} = \{\hat{w}_{j|t}\}_{j=t-M_t}^{t-1}$ and the initial state and the parameter estimates, $\hat{x}_{t-M_t|t}$ and~$\hat{\theta}_{t}$, respectively.
Using the system model~\eqref{eq:sys}, the disturbance sequence, initial state estimate, and parameter estimate define a sequence of state $\{\hat{x}_{j|t}\}_{j=t-M_t}^t$ and output $\{\hat{y}_{j|t}\}_{j=t-M_t}^{t-1}$ estimates.
The cost function is chosen as
\begin{align}
&V_{\mathrm{MHE}}^{\mathrm{x},\theta}(\hat{x}_{t-M_t|t},\hat{w}_{\cdot|t},\hat{y}_{\cdot|t},\hat{\theta}_{t},t) = \nonumber \\
& \quad 2\eta^{M_t}\left(\|\hat{x}_{t-M_t|t}-\hat{x}_{t-M_t} \|_{P_{\mathrm{x}}}^2 + \| \hat{\theta}_{t} - \hat{\theta}_{t-M_t} \|_{P_{\theta}}^2 \right) \label{eq:MHE_objective_standard} \\
&\quad +2\sum_{j=1}^{M_t}\eta^{j-1}\left(\|\hat{w}_{t-j|t}-\bar{w}_{t-j}\|_Q^2 +\|\hat{y}_{t-j|t}-\bar{y}_{t-j}\|_R^2 \right), \nonumber
\end{align}
with a possible discount factor $\eta\in [0,1]$ and weighting matrices $P_{\mathrm{x}},\ P_{\theta},\ R,\ Q\succ 0$.
The two first terms in~\eqref{eq:MHE_objective_standard} involving the past state and parameter estimates, $\hat{x}_{t-M_t}$ and $\hat{\theta}_{t-M_t}$, are called \textit{prior weighting} on the state and parameter, respectively.
The MHE objective~\eqref{eq:MHE_objective_standard} is most commonly chosen without discounting, i.e., with $\eta=1$~\cite{Robertson1996,Zanon2013,Boegli2013,Poloni2010,Abdollahpouri2017,Russo1999,Kuhl2011,Valluru2017,Kupper2009,Chen2012,Frick2012,Tuveri2023}.
The state and parameter estimates at time step $t$ are then obtained by solving the following nonlinear program (NLP)

\begin{subequations}\label{eq:MHE_state_and_parameter}
	\begin{align}\label{eq:MHE_state_and_parameter_cost}
	\min_{\hat{x}_{t-M_t|t},\hat{w}_{\cdot|t},\hat{\theta}_{t}}& V_{\mathrm{MHE}}^{\mathrm{x},\theta}(\hat{x}_{t-M_t|t},\hat{w}_{\cdot|t},\hat{y}_{\cdot|t},\hat{\theta}_{t},t) \\ \label{eq:MHE_state_and_parameter_1}
	\text{s.t. }&\hat{x}_{j+1|t}=f(\hat{x}_{j|t},u_{j},\hat{w}_{j|t},\hat{\theta}_{t}),\ j\in\mathbb{I}_{[t-M_t,t-1]},\\
	\label{eq:MHE_state_and_parameter_2}
	&\hat{y}_{j|t}=h(\hat{x}_{j|t},u_{j},\hat{w}_{j|t},\hat{\theta}_{t}),\ j\in\mathbb{I}_{[t-M_t,t-1]},\\
	\label{eq:MHE_state_and_parameter_3}
	&\hat{w}_{j|t} \in\mathbb{W},\ \hat{y}_{j|t}\in\mathbb{Y},\  j\in\mathbb{I}_{[t-M_t,t-1]}, \\
	\label{eq:MHE_state_and_parameter_4}
	&\hat{x}_{j|t}\in\mathbb{X},\ j\in\mathbb{I}_{[t-M_t,t]}, \\
	\label{eq:MHE_state_and_parameter_5}
	&\hat{\theta}_{t} \in \Theta.
	\end{align}
\end{subequations}
We denote a (not necessarily unique) minimizer of~\eqref{eq:MHE_state_and_parameter} as $\hat{x}_{t-M_t|t}^*$, $\hat{w}_{\cdot|t}^*$, and $\hat{\theta}_{t}^*$, with the corresponding state and output sequences denoted as $\hat{x}_{\cdot|t}^*$ and $\hat{y}_{\cdot|t}^*$, respectively.
The resulting state and parameter estimates at time step $t$ are denoted as
$\hat{x}_t = \hat{x}_{t|t}^*$ and $\hat{\theta}_t = \hat{\theta}_{t}^*$, respectively.
The MHE scheme~\eqref{eq:MHE_state_and_parameter} is applied in receding horizon fashion.
Thereby, the current state and parameter estimates are obtained at each time step~$t$ by solving the (finite-horizon) MHE problem~\eqref{eq:MHE_state_and_parameter} based on the $M_t$ most recent inputs, noise and output estimates, and the prior weightings based on the state and parameter estimates at time step $t-M_t$.

Stability of such an MHE scheme can be established using tools from~\cite{Schiller2023} if the augmented state $\begin{bmatrix} x_t^\top & \theta^\top \end{bmatrix}^\top$ is detectable, the horizon length $M$ is chosen sufficiently large, and if $\eta \in [0,1)$.
Notably, as shown in~\cite{Sui2011,schiller2023nonlinear}, detectability of the augmented state involves satisfaction of a PE condition, compare Remark~\ref{rem:PE} below for a formal definition of PE.
In contrast, the detectability condition considered in this paper (Assumption 1) does not depend on a PE condition, and thus holds uniformly for \emph{any} sequence of inputs.

\begin{remark}[Persistency of excitation]\label{rem:PE}
		Suppose for simplicity that the system~\eqref{eq:sys} is linearly parameterized
		\begin{subequations}\label{eq:sys_lin_param}
			\begin{align}
				x_{t+1} &= \tilde{f}(x_t,u_t) + G(x_t,u_t)\theta + E_{\mathrm{x}}w_t, \\
				y_t &= \tilde{h}(x_t,u_t,w_t),
			\end{align}
		\end{subequations}
		with $E_{\mathrm{x}} = \begin{bmatrix}I_n & 0 \end{bmatrix}\in \mathbb{R}^{n_{\mathrm{x}}\times n_{\mathrm{w}}}$. 
		Then, we have persistence of excitation~\cite{Johnstone1982} over a time window $M$ if there exists a positive constant $c > 0$ such that
		\begin{align}\label{eq:excitation_matrix}
			PE_t \coloneqq \sum_{k=t-M}^{t}G_k^\top G_k \succeq c I, \quad \forall t \in \mathbb{I}_{\ge M},
		\end{align}
		with $G_k = G(x_k,u_k)$.
		In case of state measurements, the accuracy of parameter estimates can be bounded based on the signal to noise ratio (SNR), i.e., the minimal singular value of $PE_t$~\eqref{eq:excitation_matrix} divided by the maximum singular value of the empirical process noise covariance $\Sigma_t^{\mathrm{w}} = \sum_{t-M}^t E_{\mathrm{x}}w_tw_t^\top E_{\mathrm{x}}^\top \in\mathbb{R}^{n_{\mathrm{w}}\times n_{\mathrm{w}}}$.
		In case this ratio is always sufficiently large, a standard MHE approach for state and parameter estimation can be used to simultaneously estimate states and parameters.
		However, such an approach can yield arbitrarily bad estimates in case the SNR is small for some time, compare the examples in Sections~\ref{sec:motivating_example} and~\ref{sec:num}.
		Note that in general~\eqref{eq:excitation_matrix} cannot be verified since only noisy output measurements are available.
\end{remark}
\section{Motivating Example}\label{sec:motivating_example}

\begin{figure}[t]
	\includegraphics[width=\columnwidth]{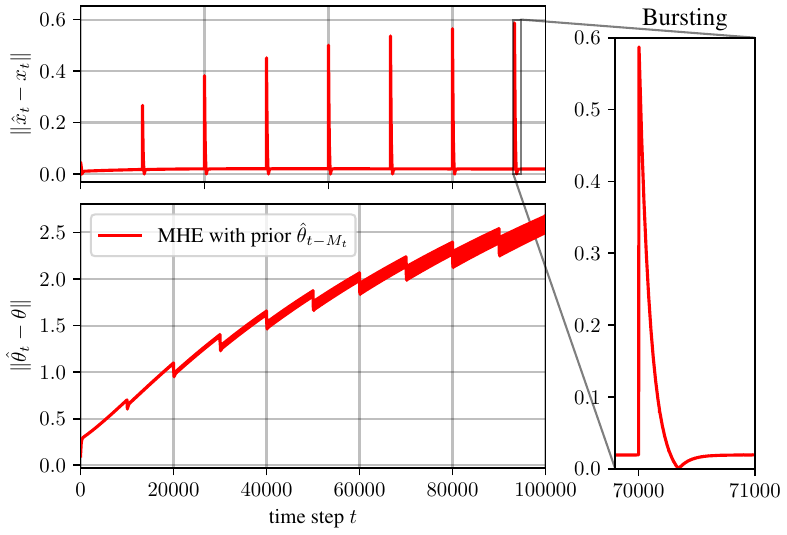}
	\caption{Academic example (Section~\ref{sec:motivating_example}): Absolute state and parameter estimation errors resulting from a standard MHE approach for joint state and parameter estimation~\eqref{eq:MHE_state_and_parameter}.}
	\label{fig:motivating_example}
\end{figure}

In this section, we provide an academic example to demonstrate that standard MHE formulations for joint state and parameter estimation (Section~\ref{sec:standard_MHE}) lack fundamental robust stability properties in the absence of PE.
The presented example is inspired by~\cite[Sec.~V]{Sethares1986}, which demonstrates parameter drift in adaptive filtering.
Consider the system
\begin{subequations}\label{eq:sys_academic}
	\begin{align}
	x_{t+1} =& a x_t + u_t + [w_t]_1, \\
	y_t =& \begin{bmatrix}
	x_t + [w_t]_2 \\
	\theta x_t + [w_t]_3
	\end{bmatrix},
	\end{align}
\end{subequations}
with $a=0.99$, $[w_t]_1 = 0$, $[w_t]_2 = -0.1$, and $[w_t]_3 = 0.1$ for all $t\in\mathbb{I}_{\ge 0}$.
Thereby, the constant values of $[w_t]_2$ and $[w_t]_3$ can be interpreted as sensor biases, which often occur in practice.
Further consider $x_0 = 1$, $\bar{x}_0 = 1$, $\theta = 1$, and $\bar{\theta}_0 = 1.1$, and the MHE formulation~\eqref{eq:MHE_state_and_parameter} with $M=20$, $\eta=0.99$, $Q=I_{3}$, $R=I_{2}$, $P_{\mathrm{x}}=10$, $P_\theta = 10$, and $\bar{w}_t = 0$ and $\bar{y}_t = y_t$ for all $t\in\mathbb{I}_{\ge 0}$.
Note that
\begin{align*}
	x_{t} = a\left([y_{t-1}]_1-[w_{t-1}]_2\right) + u_{t-1} + [w_{t-1}]_1,
\end{align*}
implies that a stable state estimator exists and hence that system~\eqref{eq:sys_academic} admits an exponential $\delta$-IOOS Lyapunov function according to Definition~\ref{def:dIOOS_Lyap} with any $\eta\in (0,1)$ (compare Proposition~\ref{prop:existence_dIOOS_Lyap}).
The input $u_t$ is chosen as
\begin{align*}
	u_t = \begin{cases}
	1 & t\mod 10000 = 0, \\
	0 & \text{else}.
	\end{cases}
\end{align*}
The input is zero over long periods of time and the state asymptotically converges to zero.
Hence, the ratio of the excitation signal compared to the noise becomes very small.
As shown in Figure~\ref{fig:motivating_example}, the parameter diverges for such a standard MHE formulation\footnote{Note that all MHE approaches applied to this motivating example are implemented in Python using CasADi~\cite{andersson2019casadi} and the solver IPOPT~\cite{Waechter2005}.} for joint state and parameter estimation.
Consequently, whenever the output contains informative data (when the input is equal to~$1$), the drifted parameter estimate leads to an increasing state estimation error.
This phenomena is known as bursting in classical adaptive control~\cite{Anderson1985,Middleton1988} and is depicted in the zoom-in in Figure~\ref{fig:motivating_example}.
Thus, robust stability of the resulting state and parameter estimate is \emph{not} guaranteed.
Note that this problem equally occurs if an MHE scheme~\eqref{eq:MHE_state_and_parameter} without discounting ($\eta = 1$) is applied.
In the next section, we show how to overcome this issue by replacing the prior weighting with a regularization on the (constant) prior parameter estimate $\bar{\theta}_0$.
In particular, we follow ideas from classical adaptive control to prevent ill-conditioned parameter estimation problems, compare our discussion in Section~\ref{sec:discussion_adaptive} below.
\section{Moving Horizon Estimation in the Absence of Persistency of Excitation} \label{sec:estimation}
In this section, we introduce an MHE formulation (Section~\ref{sec:MHE_formulation}) to estimate the state of a nonlinear dynamical system subject to parametric uncertainty in the absence of~PE.
Furthermore, we prove that the proposed MHE is $\delta$-IOpS according to Definition~\ref{def:dIOpS} (Section~\ref{sec:theoretical_analysis}), which ensures that the stability issue in the absence of PE, as discussed in the previous section, is overcome (Section~\ref{sec:motivating_example_revisited}).

\subsection{Moving Horizon Estimation Formulation}\label{sec:MHE_formulation}
Equivalently to the standard MHE formulation for joint state and parameter estimation stated in Section~\ref{sec:standard_MHE}, our proposed MHE at time step~$t$ considers past input data $\{u_j\}_{j=t-M_t}^{t-1}$ and estimates of past disturbances and outputs $\{\bar{w}_j,\bar{y}_j\}_{j=t-M_t}^{t-1}$ in a  window of length\footnote{The use of the window $M_t$ allows us to consider both the initial phase with $t\in\mathbb{I}_{[0,M-1]}$, where we rely on all available measurements, and $t\in \mathbb{I}_{\geq M}$, where we use a fixed horizon $M$.} $M_t = \min\{t,M\}$, with horizon $M \in \mathbb{I}_{\ge 0}$, and the past state estimate $\hat{x}_{t-M_t}$.
However, instead of using the past parameter estimate $\hat{\theta}_{t-M_t}$, the proposed MHE formulation relies on an available prior $\bar{\theta}_0$ of the unknown parameter.
The objective of the optimization-based state estimation problem is chosen as
\begin{align}
		&V_{\mathrm{MHE}}^{\mathrm{x}}(\hat{x}_{t-M_t|t},\hat{w}_{\cdot|t},\hat{y}_{\cdot|t},\hat{\theta}_{t},t) = \nonumber \\
		& \quad 2\eta^{M_t}\left(\|\hat{x}_{t-M_t|t}-\hat{x}_{t-M_t} \|_{P_{2,\mathrm{x}}}^2 + \| \hat{\theta}_{t} - \bar{\theta}_0 \|_{P_{2,\theta}}^2 \right) \label{eq:MHE_objective} \\
		&\quad +2\sum_{j=1}^{M_t}\eta^{j-1}\left(\|\hat{w}_{t-j|t}-\bar{w}_{t-j}\|_Q^2 +\|\hat{y}_{t-j|t}-\bar{y}_{t-j}\|_R^2 \right), \nonumber
\end{align}
where $\eta\in [0,1)$.
The discount factor $\eta$ and weighting matrices $P_{2,\mathrm{x}}$, $P_{2,\theta}$, $R$, and $Q$ are chosen based on the exponential $\delta$-IOOS Lyapuonv function according to Definition~\ref{def:dIOOS_Lyap}, which is guaranteed to exist due to Assumption~\ref{ass:existence_of_estimator} and Proposition~\ref{prop:existence_dIOOS_Lyap}.
In Remark~\ref{rem:MHE_objective} below, we discuss how to design the cost~\eqref{eq:MHE_objective} in a practical setting where a $\delta$-IOOS Lyapunov function is not known.
The state estimate at time step $t$ is then obtained by solving the following NLP

\begin{subequations}\label{eq:MHE_parametric}
	\begin{align}\label{eq:MHE_parametric_cost}
	\min_{\hat{x}_{t-M_t|t},\hat{w}_{\cdot|t},\hat{\theta}_t}& V_{\mathrm{MHE}}^{\mathrm{x}}(\hat{x}_{t-M_t|t},\hat{w}_{\cdot|t},\hat{y}_{\cdot|t},\hat{\theta}_{t},t) \\ \label{eq:MHE_parametric_1}
	\text{s.t. }&\eqref{eq:MHE_state_and_parameter_1},\eqref{eq:MHE_state_and_parameter_2},\eqref{eq:MHE_state_and_parameter_3},\eqref{eq:MHE_state_and_parameter_4},\eqref{eq:MHE_state_and_parameter_5}.
	\end{align}
\end{subequations}
We denote a (non-unique) minimizer\footnote{We assume that a minimizer to~\eqref{eq:MHE_state_and_parameter} always exists. Existence of such a minimizer can be ensured under mild assumptions, e.g., if $P_{2,\mathrm{x}}, P_{2,\theta}, Q, R\succ 0$, or the sets $\mathbb{W},\mathbb{Y},\mathbb{X},\Theta$ are compact, cf.~\cite[App.~A.11]{Rawlings2020}.\label{foot:minimier}} of~\eqref{eq:MHE_parametric} as $\hat{x}_{t-M_t|t}^*$, $\hat{w}_{\cdot|t}^*$, and $\hat{\theta}_{t}^*$, with the corresponding state and output sequences denoted as $\hat{x}_{\cdot|t}^*$ and $\hat{y}_{\cdot|t}^*$, respectively.
The resulting state estimate at time step~$t$ is denoted as
\begin{align}
	\hat{x}_t = \hat{x}_{t|t}^*, \label{eq:MHE_estimate}
\end{align}
and the corresponding estimation error as
\begin{align}
	\hat{e}_t = \hat{x}_t - x_t. \label{eq:MHE_estimate_error}
\end{align}
The MHE approach~\eqref{eq:MHE_parametric} is then applied in a receding horizon fashion.
Thereby, the state estimate~\eqref{eq:MHE_estimate} at each time step $t$ is obtained by solving the optimization problem~\eqref{eq:MHE_parametric} based on the $M_t$ most recent inputs, noise and output estimates, and the prior weighting based on the state estimate at time step~$t-M_t$.
 
\subsection{Theoretical Analysis}\label{sec:theoretical_analysis}
In the following, we establish that the MHE approach~\eqref{eq:MHE_parametric} is an exponentially $\delta$-IOpS estimator according to Definition~\ref{def:dIOpS} provided the horizon $M$ is chosen larger than some lower bound.
\begin{theorem}[MHE is $\delta$-IOpS]\label{thm:MHE_dIOpS}
	Let Assumption~\ref{ass:existence_of_estimator} hold and suppose the horizon $M$ satisfies
	\begin{align}\label{eq:M_cond}
		\rho^M \coloneqq 4\eta^{M}\lambda_{\max}(P_{2,\mathrm{x}},P_{1}) < 1,
	\end{align}
	with $\rho \in [0,1)$.
	Then, the MHE estimator~\eqref{eq:MHE_parametric} is exponentially $\delta$-IOpS according to Definition~\ref{def:dIOpS}, i.e,~\eqref{eq:ISpS_max_form} holds for all $t \in \mathbb{I}_{\ge 0}$ with (uniform) constants $C_1,\ C_2,\ C_3 > 0$, and
	\begin{align} \label{eq:MHE_ISpS_eps_lambda}
		\epsilon =&\ \frac{4}{\sqrt{1-\rho}}\sqrt{\frac{\lambda_{\max}(P_{2,\theta})}{\lambda_{\min}(P_1)}}\sqrt{\rho}^M\| \bar{\theta}_0 - \theta \|, \\
		\lambda_1 =&\ \sqrt{\rho},\ \lambda_2 = \lambda_3 = \sqrt[4]{\rho}. \nonumber
	\end{align}
\end{theorem}

The proof of Theorem~\ref{thm:MHE_dIOpS} can be found in Appendix~\ref{app:proof_thm_dIOpS}.
As a consequence of~\eqref{eq:MHE_ISpS_eps_lambda}, the resulting bound on the state estimation error~\eqref{eq:ISpS_max_form} improves with increasing horizon length~$M$ and in case the prior estimate $\bar{\theta}_0$ is closer to the true (unknown) parameter $\theta$.
Note that choosing the horizon length $M$ poses a trade-off between computational complexity and accuracy, while robustness is preserved provided Condition~\eqref{eq:M_cond} is satisfied.
For $M$ going to infinity, it follows that the MHE approach in~\eqref{eq:MHE_parametric} is an exponentially $\delta$-IOS estimator according to Definition~\ref{def:dIOpS}, even for arbitrarily bad prior parameter estimates~$\bar{\theta}_0$.
Consequently, FIE, i.e., Problem~\eqref{eq:MHE_parametric} with $M_t=t$, is exponentially $\delta$-IOS.
\begin{corollary}[FIE is $\delta$-IOS]\label{cor:FIE_dIOS}
	Let Assumption~\ref{ass:existence_of_estimator} hold. Then, the FIE estimator~\eqref{eq:MHE_parametric} with $M_t=t$ is exponentially $\delta$-IOS according to Definition~\ref{def:dIOpS}.
\end{corollary}
\begin{proof}
	The proof follows directly from~\eqref{eq:MHE_dIOpS_bound} by setting $M=t$.
\end{proof}

The above corollary shows that the results from~\cite[Thm.~9]{Muntwiler2023} are recovered.

\begin{remark}[Known parameter]
	In case the true parameter is known, i.e., $\bar{\theta}_0 = \theta$, it follows directly from Theorem~\ref{thm:MHE_dIOpS}, and in particular~\eqref{eq:MHE_ISpS_eps_lambda}, that the estimation error~\eqref{eq:MHE_estimate_error} resulting from the proposed MHE approach~\eqref{eq:MHE_parametric} is robustly globally exponentially stable (with $\epsilon = 0$) according to~\cite[Def.~1]{knufer2018robust}.
\end{remark}
\begin{remark}[Parameterization of MHE objective]\label{rem:MHE_objective}
	Choosing $P_{2,\mathrm{x}}$, $P_{2,\theta}$, $R$, and $Q$ in~\eqref{eq:MHE_objective} based on a $\delta$-IOOS Lyapunov function $W_\delta$ (Definition~\ref{def:dIOOS_Lyap}) might seem to limit the applicability of the proposed MHE design and theory since such a Lyapunov function is generally not known explicitly in practical applications.
	However, due to Assumption~\ref{ass:existence_of_estimator} and Proposition~\ref{prop:existence_dIOOS_Lyap} some $W_\delta$ with $\tilde{P}_{2,\mathrm{x}}$, $\tilde{P}_{2,\theta}$, $\tilde{R}$, $\tilde{Q}\succ 0$ and $\tilde{\eta} \in [0,1)$ is guaranteed to exist.
	Due to scale invariance of $W_\delta$, the weighting matrices $P_{2,\mathrm{x}}$, $P_{2,\theta}$, $R$, $Q\succ 0$ in the objective~\eqref{eq:MHE_objective} can be chosen arbitrarily, provided $\eta > \tilde{\eta}$, i.e., $1-\eta > 0$ is sufficiently small, $P_1 \succ 0$ is chosen sufficiently small, and the horizon $M$ is long enough (Condition~\eqref{eq:M_cond} in Theorem~\ref{thm:MHE_dIOpS}), compare also a similar discussion in~\cite[Rem.~1]{Schiller2023}.
\end{remark}
\begin{remark}[Benefits of online parameter adaptation]
	An alternative approach to prevent parameter drift in case of insufficient PE is to \emph{not} adapt/estimate the parameters online, i.e., to apply the MHE approach~\eqref{eq:MHE_parametric} with additional constraint $\hat{\theta}_t = \bar{\theta}_0$.
	However, such an approach can lead to infeasibilities and persistent state estimation errors even with small noise, compare also the numerical comparison in Section~\ref{sec:num}.
\end{remark}
\begin{remark}[Slowly varying parameters] \label{rem:varying_params}
	In practical applications the parameter $\theta$ might be slowly varying over time, e.g., due to aging of components.
	Suppose slow parameter variations, i.e., $\theta_t$ to be constant within one horizon length $M$.
	Then, the proposed MHE approach~\eqref{eq:MHE_parametric} also ensures $\delta$-IOpS w.r.t. to
	\begin{align*}
		\epsilon =&\ \frac{4}{\sqrt{1-\rho}}\sqrt{\frac{\lambda_{\max}(P_{2,\theta})}{\lambda_{\min}(P_1)}}\sqrt{\rho}^M \max_{j\in\mathbb{I}_{[0,t]}}\| \bar{\theta}_0 - \theta_j \|.
	\end{align*}
	Note that, similar to~\eqref{eq:MHE_ISpS_eps_lambda}, $\epsilon$ approaches zero for increasing horizon length $M$.
	Recently, an MHE design to jointly estimate states and time-varying parameters was presented in~\cite{schiller2024moving}, however, with stability properties depending on the level of PE.
\end{remark}
\subsection{Motivating Example Revisited}\label{sec:motivating_example_revisited}

\begin{figure}[t]
	\includegraphics[width=\columnwidth]{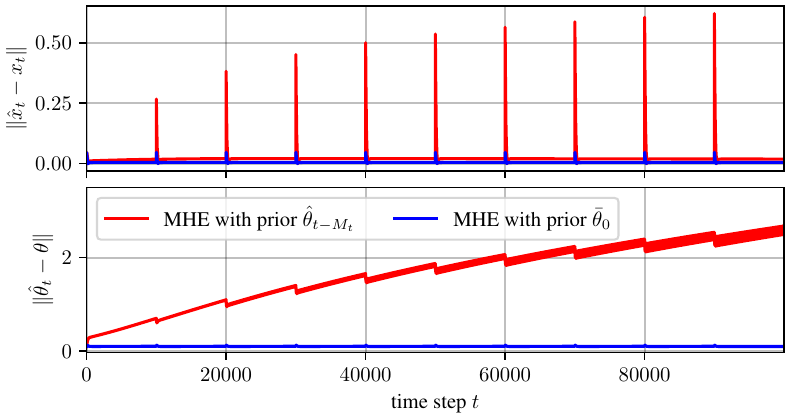}
	\caption{Academic example (Section~\ref{sec:motivating_example}): Absolute state and parameter estimation errors resulting from a standard MHE approach for joint state and parameter estimation with prior on $\hat{\theta}_{t-M_t}$~\eqref{eq:MHE_state_and_parameter} and our proposed MHE approach with prior $\bar{\theta}_0$~\eqref{eq:MHE_parametric}.}
	\label{fig:academic_example_ours}
\end{figure}

In the following, we revisit the academic example system~\eqref{eq:sys_academic} discussed in Section~\ref{sec:motivating_example} and apply the proposed MHE approach~\eqref{eq:MHE_parametric} to it.
We choose the same values for $\bar{\theta}_0$, $M$, $\eta$, $Q$, $R$ as for the standard MHE in Section~\ref{sec:motivating_example}, with $P_{2,\mathrm{x}}=P_{\mathrm{x}}$, and $P_{2,\theta}=P_\theta$.
The corresponding absolute state estimation and parameter errors are shown in Figure~\ref{fig:academic_example_ours}.
For comparison, the estimation results from Section~\ref{sec:motivating_example} using a standard MHE approach for state and parameter estimation~\eqref{eq:MHE_state_and_parameter} are shown again.
As expected from the theoretical analysis (Theorem~\ref{thm:MHE_dIOpS}), the state estimation error resulting from the proposed MHE is significantly smaller.
In particular, the parameter error remains uniformly bounded and the parameter converges to the prior $\bar{\theta}_0$ whenever the data is not informative enough.
\section{Discussion}\label{sec:discussion}

In this section, we discuss the relation of the MHE approach presented in Section~\ref{sec:estimation} to existing methods for state and parameter estimation from the literature.

\subsection{Standard Moving Horizon Estimation}\label{sec:discussion_standard_MHE}
Assuming perfect knowledge of model parameters $\theta$, $\delta$-IOSS is a necessary and sufficient condition for robustly stable state estimation~\cite{Allan2021}.
Furthermore, MHE is robustly stable assuming $\delta$-IOSS and a sufficiently long horizon \cite{Schiller2023,knuefer2021MHE}.
Consequently, a robustly stable MHE approach can be designed whenever a robustly stable state estimator exists.

MHE is often applied to estimate an augmented system state combining both states and parameters~\cite{ Zanon2013,Boegli2013,Poloni2010,Abdollahpouri2017,Russo1999,Kuhl2011,Kupper2009,Valluru2017,Chen2012,Tuveri2023,kleyman2023state,Frick2012}.
Similarly to the state estimation case, the existence of a robustly (uniformly) stable estimator for the augmented system state implies stability of MHE with sufficiently long horizon, compare~\cite{schiller2023nonlinear}.
However, existence of a stable estimator in this setting generally requires satisfaction of some PE condition (cf. Remark~\ref{rem:PE}), which is in general not given.
Lack of PE can lead to parameter drift and unstable state estimates as shown in Section~\ref{sec:motivating_example}.
We overcome this issue by replacing the standard prior weighting on the parameter with a regularization with respect to a (constant) prior parameter estimate $\bar{\theta}_0$.
Only assuming existence of a robustly stable estimator (Assumption~\ref{ass:existence_of_estimator}), the proposed MHE approach is shown to result in practically stable estimates of the state irrespective of PE (Theorem~\ref{thm:MHE_dIOpS}), i.e., the proposed MHE ``performs well'' whenever it is possible to design any observer satisfying the desired stability property.

\subsection{Adaptations of Prior Weighting in Absence of PE}\label{sec:discussion_numerical_tricks}
The prior weighting in the MHE objective~\eqref{eq:MHE_objective_standard} is often chosen such that the solution of the MHE problem is as close as possible to the infinite-horizon (FIE) solution.
In the linear unconstrained case with quadratic cost and no unknown parameters, the prior weighting can be computed using a standard Kalman filter such that the MHE solution matches exactly to the one of FIE~\cite{Rao2001}. 
Thus, the prior weighting in many nonlinear MHE formulations is computed using an extended Kalman filter~\cite{rao2003constrained,Deniz2020,Deniz2021,Deniz2019,Kuhl2011}.
In case of joint state and parameter estimation in the absence of PE, such a prior weighting becomes singular/ill-conditioned.
Thus a minimizer to the optimization problem may not exist, which leads to a lack of robustness both from a robust and statistical point of view~\cite{Cao2022}.
In~\cite{Baumgartner2022}, a regularization of the prior weighting in connection with pseudo measurements was introduced, thus reducing the sensitivity to measurement outliers, compare also~\cite{Sui2011} for a similar approach.
These mechanisms are comparable to forgetting factors in classical adaptive estimation and control~\cite{goodwin2014adaptive}.
However, in contrast to the proposed MHE scheme, the approaches in~\cite{Sui2011,Baumgartner2022} with adapted prior weightings lack (robust) stability guarantees in the absence of PE.

More recently, an MHE approach for joint state and parameter estimation with online verification of PE has been introduced in~\cite{schiller2023nonlinear}.
Thereby, the prior parameter estimate is kept fixed whenever there is insufficient PE as also done in the MHE approach proposed in Section~\ref{sec:MHE_formulation}.
The MHE approach in~\cite{schiller2023nonlinear} offers the possibility to adapt the prior parameter estimate in case of sufficient PE.
However, practical application of the online verification of PE relies on a stricter assumption~\cite[Ass.~12]{schiller2023nonlinear}, which limits applicability to weakly nonlinear systems with small parameter errors.
In contrast, in this paper we rely on a necessary and sufficient condition (Ass.~\ref{ass:existence_of_estimator}) to ensure practically robustly stable state estimates, even in cases of arbitrarily large initial parameter estimation errors.

\subsection{Adaptive Estimation}\label{sec:discussion_adaptive}

\begin{figure}[t]
	\includegraphics[width=\columnwidth]{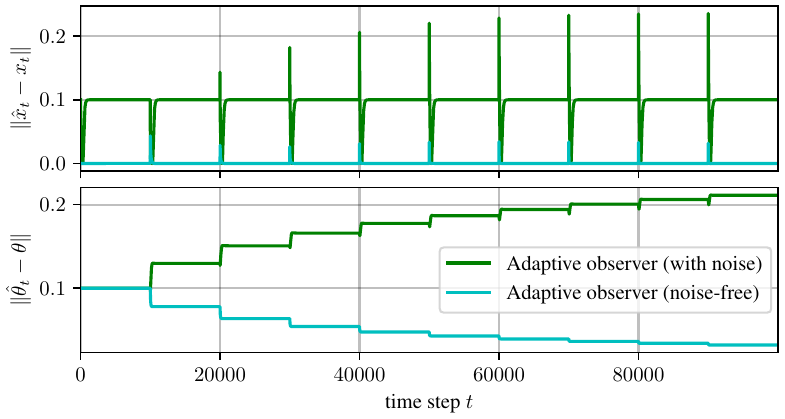}
	\caption{Academic example (Section~\ref{sec:motivating_example}): Absolute state and parameter estimation errors resulting from the adaptive observer proposed in~\cite{Dey2022} for both the noise-free case ($w=0$) and the case with measurement noise.}
	\label{fig:academic_example_adaptive}
\end{figure}

Classically, adaptive observers are applied for state estimation in case of parametric uncertainty.
While there exist adaptive observers with guarantees in the absence of PE~\cite{Marino2001,Tomei2022,Dey2022}, those approaches in general lack robustness against process and measurement noise.
In Figure~\ref{fig:academic_example_adaptive}, we show the application of the adaptive estimation approach\footnote{Application of the adaptive observer from~\cite{Dey2022} requires to transform the system~\eqref{eq:sys_academic} to observable canonical form, consequently having a state of dimension $n_{\mathrm{x}}=2$. We choose the parameters of the observer from~\cite{Dey2022} as $F= a \cdot I_2$ (with $a=0.99$ from the dynamics introduced in Section~\ref{sec:motivating_example}), $R=10\cdot I_2$, and apply covariance resetting~\cite[Sec.~3.1]{Dey2022} with $k_0 = 0.1$ and $k_{\min} = 0.01$.} from~\cite{Dey2022} to the academic example presented in Section~\ref{sec:motivating_example}, which is characterized by measurement noise and absence of PE.
It is visible that the parameter estimate converges in the noise-free case (with $w=0$), resulting in small state estimation errors.
In the noisy case, we observe a similar effect as for the MHE with prior on $\hat{\theta}_{t-M_t}$ as shown in Figure~\ref{fig:academic_example_ours}.
In particular, the slowly diverging parameter estimates lead to large state estimation errors whenever there is enough excitation.
In contrast, the proposed MHE approach is practically robustly stable (Definition~\ref{def:dIOpS}) for any general nonlinear system whenever there exists \emph{any} robustly stable estimator (Assumption~\ref{ass:existence_of_estimator}).
Hence, whenever it is in principle possible to design an adaptive observer with robust stability properties, then, Assumption~\ref{ass:existence_of_estimator} holds and the proposed MHE formulation yields practical robust stability properties.
However, to the best of the author's knowledge, existing adaptive observer designs do not have such properties for general nonlinear systems.

The problem of parameter drift and the related bursting phenomena is well known in the adaptive control and estimation literature~\cite{Rohrs1982,Middleton1988,Sethares1986,Anderson1985} and the interplay between PE and robustness remains actively researched~\cite{Johnstone1982,Wang2020}.
We propose a simple constructive solution to this problem by introducing a regularization cost on the parameter estimate w.r.t. a constant initial parameter estimate $\bar{\theta}_0$.
As shown in~\cite{Ohmori1989}, such a regularization can be related to the classical $\sigma$-modification, compare~\cite[Chap.~8]{narendra2012stable},~\cite{ioannou1986robust}, which ensures boundedness/robustness irrespective of PE.

\subsection{Verification of Assumption~\ref{ass:existence_of_estimator} in Practical Applications} \label{sec:verify_ass}

Our theoretical results in Section~\ref{sec:theoretical_analysis} only require the existence of a robustly stable state estimator (Assumption~\ref{ass:existence_of_estimator}). 
While this assumption is clearly necessary, it is nontrivial to verify whether it holds in practical applications. 
In the following, we show that the existence of a stable estimator (Assumption~\ref{ass:existence_of_estimator}) can be established for a very common class of (mechanical) systems.
In particular, consider a discrete-time system of the form
\begin{align}\label{eq:sys_q}
	\begin{bmatrix}
	p_{t+1} \\ v_{t+1} \\ y_t
	\end{bmatrix} = \begin{bmatrix}
	p_t + \Delta t A(p_t)v_t \\ v_t + \Delta t f(p_t,v_t,u_t,\theta) \\
	p_t
	\end{bmatrix} + w_t,
\end{align}
with state $\begin{bmatrix}p_t^\top & v_t^\top\end{bmatrix}^\top\in\mathbb{R}^{2\cdot n_{\mathrm{p}}}$, control input $u_t\in\mathbb{R}^{n_{\mathrm{u}}}$, parameter $\theta\in\mathbb{R}^{n_{\theta}}$, measurement $y_t\in\mathbb{R}^{n_{\mathrm{p}}}$, process and measurement noise $w_t\in\mathbb{R}^{3\cdot n_{\mathrm{p}}}$, sampling time $\Delta t$, and $A^{-1}(p)$ Lipschitz continuous.
Such a model~\eqref{eq:sys_q} can describe the behavior of many mechanical systems where the position $p$ can be measured, e.g., quadrotors~\cite[Sec.~V.B]{Schiller2023}, and cars~\cite{Carron2022}, compare also the presented numerical example in Section~\ref{sec:num} below.

In the noise-free case (i.e., with $w=0$), the state of system~\eqref{eq:sys_q} can be exactly recovered given two past outputs~\eqref{eq:car_meas}, even without knowledge of the parameters.
Hence, a robustly stable state estimator can be constructed, i.e., Assumption~\ref{ass:existence_of_estimator} is satisfied\footnote{In the noise-free case, the states can be recovered as $p_t = y_t$, $v_t = \frac{1}{\Delta t}A^{-1}(y_t)\left(y_{t+1} - y_t\right)$.
Due to Lipschitz continuity of $A^{-1}$, the state estimation error is bounded proportionally to the magnitude of the process and measurement noise, i.e., Equation~\eqref{eq:ISpS_max_form} holds with arbitrarily small positive constants $\lambda_1, \lambda_2, \lambda_3 > 0$ and $\epsilon = 0$.
However, note that this estimator is \emph{not} causal, which slightly deviates from Definition~\ref{def:state_estimator}.}.
On the other hand, the augmented state $\begin{bmatrix}p_t^\top & v_t^\top & \theta^\top \end{bmatrix}^\top$ of system~\eqref{eq:sys_q} is in general not uniformly detectable, compare, e.g., the numerical example in Section~\ref{sec:num} below.
\section{Practical Application Example}\label{sec:num} 

In this section, we present a numerical example of state estimation under parametric uncertainty for a car in a realistic scenario.
We show that a standard MHE approach for joint state and parameter estimation results in detrimental state estimates due to parameter drift, while the proposed MHE approach ensures robust state estimation.
In addition, we show that the proposed adaptation allows for superior state estimation compared to an MHE with fixed parameter value and consequently a persistent parameter error.
 
\subsection{State Estimation for Car}
We consider a dynamic bicycle model with simplified Pacejka tire force model~\cite{rajamani2011vehicle,Carron2022}.
The continuous-time system is governed by the following differential equations
\begin{subequations}\label{eq:dynamic}
	\begin{align}
		\dot{x}_{\mathrm{p}} &= v_{\mathrm{x}} \cos(\psi) - v_{\mathrm{y}} \sin(\psi),\\
		\dot{y}_{\mathrm{p}} &= v_{\mathrm{x}} \sin(\psi) + v_{\mathrm{y}} \cos(\psi),\\
		\dot{\psi} &= \omega,\\
		\dot{v}_{\mathrm{x}} &= \frac{1}{m} \left( F_{\mathrm{x}} - F_{\mathrm{f}} \sin(\delta) + mv_{\mathrm{y}} \omega\right), \\
		\dot{v}_{\mathrm{y}} &= \frac{1}{m} \left( F_{\mathrm{r}} + F_{\mathrm{f}} \cos(\delta) - mv_{\mathrm{x}} \omega\right), \\
		\dot{\omega} &= \frac{1}{I_{\mathrm{z}}} \left( F_{\mathrm{f}} l_{\mathrm{f}} \cos(\delta) - F_{\mathrm{r}} l_{\mathrm{r}}\right),
	\end{align}
\end{subequations}%
where $x_{\mathrm{p}}$, $y_{\mathrm{p}}$, $\psi$ are the x and y coordinates of the car position and the yaw angle in world frame, $\omega$ is the yaw rate, and $v_{\mathrm{x}}$ and $v_{\mathrm{y}}$ are the longitudinal and lateral velocities in body frame of the car. Further, $m$ is the mass of the car, $I_{\mathrm{z}}$ is the inertia along the $z$-axis, and $l_{\mathrm{f}}$ and $l_{\mathrm{r}}$ are the distance of the front and rear axis from the center of mass, respectively. The lateral tire forces $F_{\mathrm{f}}$ and $F_{\mathrm{r}}$ are modeled with the simplified Pacejka tire force model
\begin{subequations} \label{eq:dynamic_bicycle_lateral_force}
	\begin{align}
		\alpha_{\mathrm{f}} &= \arctan\left( \frac{v_{\mathrm{y}} + \omega l_{\mathrm{f}}}{v_{\mathrm{x}}} \right) - \delta, \label{eq:alpha_f}\\
		\alpha_{\mathrm{r}} &= \arctan\left( \frac{v_{\mathrm{y}} - \omega l_{\mathrm{r}}}{v_{\mathrm{x}}} \right), \label{eq:alpha_r}\\
		F_{\mathrm{f}} &= D_{\mathrm{f}} \sin(C_{\mathrm{f}} \arctan(B_{\mathrm{f}} \alpha_{\mathrm{f}})), \label{eq:F_f} \\
		F_{\mathrm{r}} &= D_{\mathrm{r}} \sin(C_{\mathrm{r}} \arctan(B_{\mathrm{r}} \alpha_{\mathrm{r}})), \label{eq:F_r}
	\end{align}
\end{subequations}
where $\alpha_{\mathrm{f}}$ and $\alpha_{\mathrm{r}}$ are the front and rear slip angles, and $B_{\mathrm{f}},\ B_{\mathrm{r}},\ C_{\mathrm{f}},\ C_{\mathrm{r}},\ D_{\mathrm{f}}$, and $D_{\mathrm{r}}$ are the Pacejka tire model parameters. 
The control input is $u=\begin{bmatrix} \delta & F_{\mathrm{x}}\end{bmatrix}^\top \in \mathbb{R}^2$, where $\delta$ is the steer angle of the front wheels and $F_{\mathrm{x}}$ the longitudinal force generated by the motor, and the overall state is $x = \begin{bmatrix} x_{\mathrm{p}} & y_{\mathrm{p}} & \psi & v_{\mathrm{x}} & v_{\mathrm{y}} & \omega\end{bmatrix}^\top \in \mathbb{X} = \mathbb{R}^6$.
The parameters of the considered system~\eqref{eq:dynamic} are stated in Table~\ref{tab:car_params} and represent a miniature race car from~\cite{Carron2022}.
The Pacejka parameters $D_{\mathrm{f}}$ and $D_{\mathrm{r}}$ are considered to be unknown, i.e., $\theta = \begin{bmatrix} D_{\mathrm{f}} & D_{\mathrm{r}} \end{bmatrix}^\top \in \Theta = \mathbb{R}^2$, with given prior $\bar{\theta}_0 = 2 \cdot \theta$.
The system dynamics~\eqref{eq:dynamic} are discretized using Euler forward with sampling time $\Delta t = 10ms$, and the resulting discrete-time model is subject to uniformly distributed additive process noise $w_{\mathrm{x}}\in\mathbb{W}_{\mathrm{x}}=\{w_{\mathrm{x}}\in\mathbb{R}^6|\|w_{\mathrm{x}}\|_\infty\leq 0.01\}$.

\begin{table}[!b]
	\centering
	\caption{Parameters of miniature car from~\cite{Carron2022}}
	\label{tab:car_params}
	\begin{tabular}{c | c | c | c  }
		$l_{\mathrm{r}}$ & $l_{\mathrm{f}}$ & $m$ & $I_{\mathrm{z}}$ \\
		\hline
		$0.038$ & $0.052$ & $0.181$ & $5.05\cdot 10^{-4}$ 
	\end{tabular} \\ \vspace{2mm}
	\begin{tabular}{ c | c | c | c | c | c }
		$B_{\mathrm{r}}$ & $C_{\mathrm{r}}$ & $D_{\mathrm{r}}$ & $B_{\mathrm{f}}$ & $C_{\mathrm{f}}$ & $D_{\mathrm{f}}$ \\
		\hline
		$8.5$ & $1.45$ & $-1.0$ & $5.2$ & $1.5$ & $-0.65$  
	\end{tabular}
\end{table}

At each time step, we obtain measurements of the position ($x_{\mathrm{p}}, y_{\mathrm{p}}$) and yaw angle $\psi$ of the car, i.e.,
\begin{equation}\label{eq:car_meas}
	y = \begin{bmatrix}
	x_{\mathrm{p}} &
	y_{\mathrm{p}} &
	\psi
	\end{bmatrix}^\top + w_{\mathrm{y}},
\end{equation}
which are subject to uniformly distributed additive measurement noise $w_{\mathrm{y}}\in\mathbb{W}_{\mathrm{y}}=\{w_{\mathrm{y}}\in\mathbb{R}^3|\ | w_{\mathrm{y}} | \le \begin{bmatrix} 0.2 & 0.2 & 0.01 \end{bmatrix}^\top\}$.
Such measurements can be provided by GPS in an outdoor environment or camera-based systems in an indoor environment, compare, e.g.,~\cite{Carron2022}.
The overall process and measurement noise is denoted as $w=\begin{bmatrix}w_{\mathrm{x}}^\top & w_{\mathrm{y}}^\top\end{bmatrix}^\top\in\mathbb{W}=\mathbb{W}_{\mathrm{x}}\times\mathbb{W}_{\mathrm{y}}\subseteq\mathbb{R}^9$.

Note that the augmented state $\begin{bmatrix}x^\top & \theta^\top \end{bmatrix}^\top$ is not uniformly detectable as the friction parameters cannot be estimated for low/zero slip $\alpha_{\mathrm{f}}$ and $\alpha_{\mathrm{r}}$.
In particular, if the car is driving straight ($\omega = 0$, $\delta = 0$, $v_{\mathrm{y}} = 0$), the slip angles~\eqref{eq:alpha_f}-\eqref{eq:alpha_r} and tire forces~\eqref{eq:F_f}-\eqref{eq:F_r} are zero, i.e., no information about $\theta$ can be inferred.
Consequently, \emph{no} robustly stable estimator for the augmented state exists.
On the other hand, the system~\eqref{eq:dynamic}-\eqref{eq:car_meas} is in the form of~\eqref{eq:sys_q} and hence the arguments in Section~\ref{sec:verify_ass} apply, i.e., the derived stability results are applicable.

\subsection{Numerical Estimation Results}

\begin{figure}[t!]
	\includegraphics[width=\columnwidth]{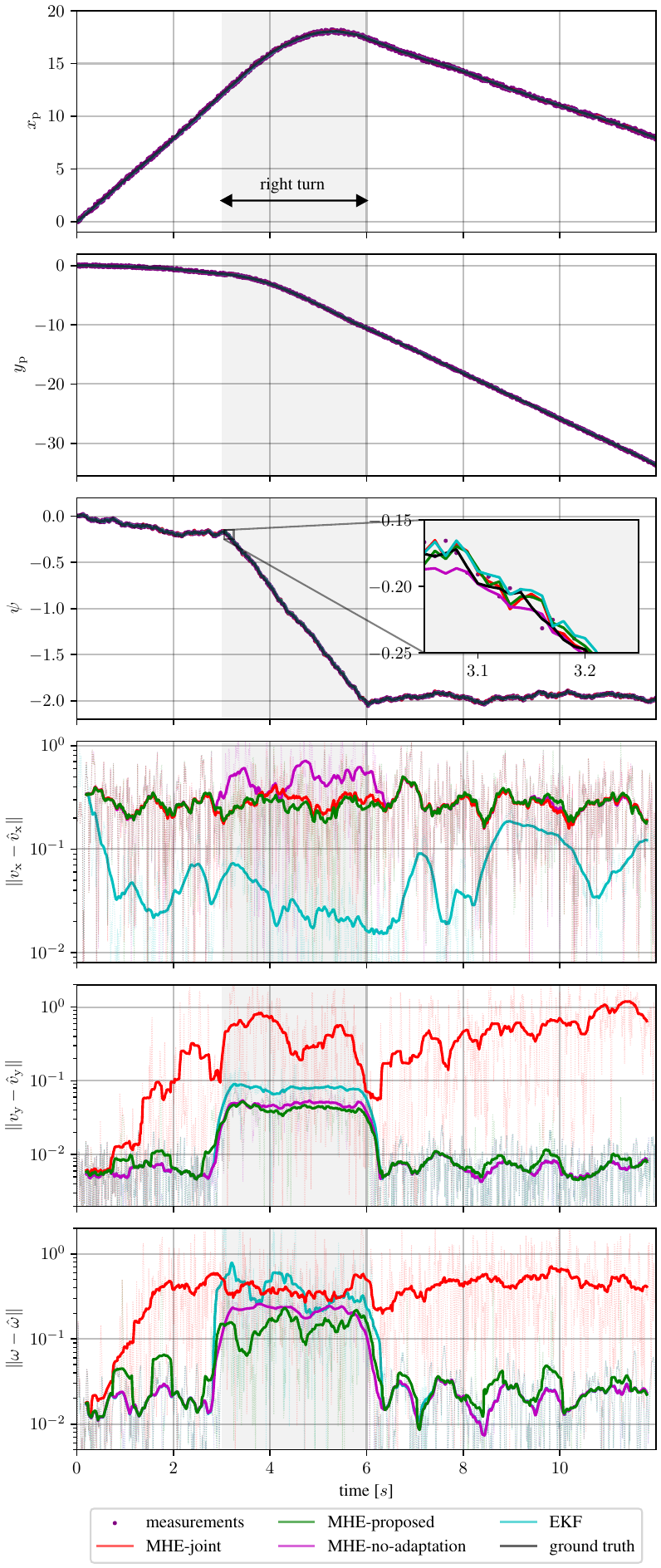}
	\caption{Miniature race car example (Section~\ref{sec:num}): The top three plots show the states $x_{\mathrm{p}},\ y_{\mathrm{p}},\ \psi$ with corresponding sensor measurements and estimates resulting from all three considered MHE approaches and an EKF. The bottom three subplots show the estimation error for $v_{\mathrm{x}},\ v_{\mathrm{y}},\ \omega$ resulting from all three MHE approaches and the EKF (dotted line) and a moving average thereof over $40$ time steps (solid line).}
	\label{fig:car_states}
\end{figure}

\begin{figure}[t!]
	\includegraphics[width=\columnwidth]{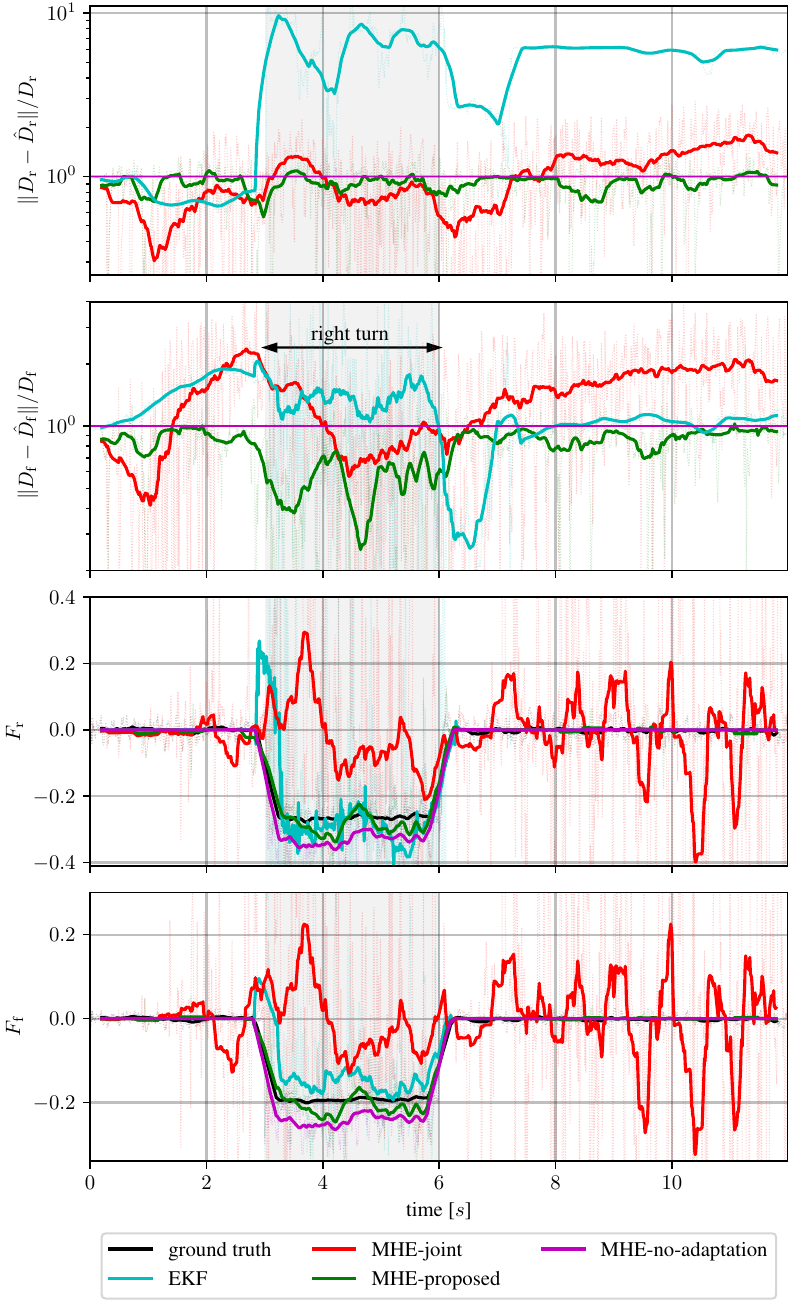}
	\caption{Miniature race car example (Section~\ref{sec:num}): The top two subplots show the normalized parameter estimation errors resulting from all three considered MHE approaches and an EKF (dotted line) and a moving average thereof over $40$ time steps (solid line).
	The bottom two subplots show the tire forces $F_{\mathrm{r}}$ and $F_{\mathrm{f}}$ and the corresponding estimates thereof resulting from all three MHE approaches and the EKF.}
	\label{fig:car_params}
\end{figure}

\begin{figure}[t!]
	\includegraphics[width=\columnwidth]{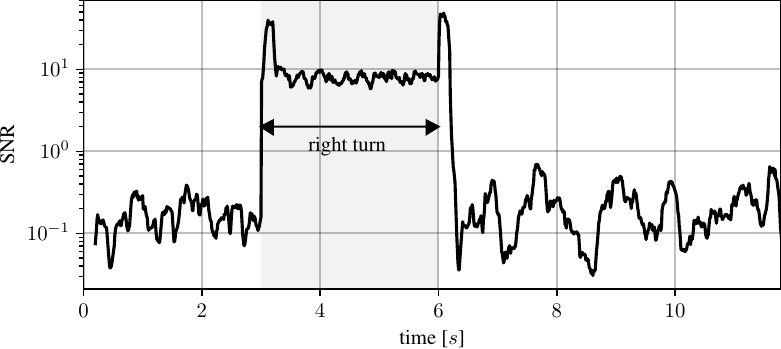}
	\caption{Miniature race car example (Section~\ref{sec:num}): Signal to noise ratio (SNR) of minimum singular value of excitation matrix $PE_t$ divided by the maximum singular value of the empirical process noise covariance $\Sigma_t^{\mathrm{w}}$, compare Remark~\ref{rem:PE}.}
	\label{fig:car_exitation}
\end{figure}

In the following, we consider a scenario where the car is driving straight ($u_t = \begin{bmatrix} 0 & 0.01 \end{bmatrix}^\top$), takes a right turn between time $3s$ and $6s$ ($u_t = \begin{bmatrix} -0.03 & 0.01 \end{bmatrix}^\top$), and drives straight again, see the top three subplots of Figure~\ref{fig:car_states}.
Note that in this scenario, the signal to noise ratio (cf. Remark~\ref{rem:PE}) is small before and after the turn, while it is large during the turn, compare Figure~\ref{fig:car_exitation}.

We compare the following four approaches: 
\begin{enumerate}
	\item \textit{MHE-proposed}: the proposed MHE approach with prior weighting w.r.t. $\bar{\theta}_0$ as defined in Section~\ref{sec:estimation},
	\item \textit{MHE-joint}: a standard MHE approach for joint state and parameter estimation using updated prior $\hat{\theta}_{t-M_t}$ as stated in Section~\ref{sec:standard_MHE},
	\item \textit{MHE-no-adaptation}: an MHE approach for state estimation without parameter adaptation, i.e., the proposed MHE formulation but with fixed parameter estimate $\hat{\theta}_t = \bar{\theta}_0$,
	\item \textit{EKF}: an EKF to estimate the extended state $\tilde{x}_t = \begin{bmatrix} x_t^\top & \theta_t^\top\end{bmatrix}^\top \in \mathbb{R}^8$~\cite{Ljung1979}.
\end{enumerate}
All three MHE approaches are implemented in Python using CasADi~\cite{andersson2019casadi} and the solver IPOPT~\cite{Waechter2005}.
The code is available\footnote{\href{https://gitlab.ethz.ch/ics/parametric-mhe}{https://gitlab.ethz.ch/ics/parametric-mhe}} online.
For the MHE approaches, we choose $Q = \mathrm{diag}(10^{4}\cdot I_6, 25\cdot I_2, 10^{4})$ and $R = \mathrm{diag}(25\cdot I_2, 10^{4})$, reflecting the bounds in $\mathbb{W}$.
Furthermore, we choose $\eta = 0.9$,  $P_{\mathrm{x}} = P_{2,\mathrm{x}} = I_6 $, $P_{\theta} = P_{2,\theta} = I_2$, a horizon $M=20$, and $\bar{w}_t = 0$ and $\bar{y}_t = y_t$ for all $t\in\mathbb{I}_{\ge 0}$.
For the EKF, we choose the covariance matrices as $Q_{\mathrm{EKF}}=\mathrm{diag}(10^{-4}\cdot I_6,2\cdot I_2)$, $R_{\mathrm{EKF}} = R^{-1}/4$, and initialize the estimation error covariance matrix with $P_0 = I_8$.
The system is initialized with $x_0 = \begin{bmatrix} 0 & 0 & 0 & 4 & 0 & 0 \end{bmatrix}^\top$, while the initial estimate is chosen as $\bar{x}_0 = \begin{bmatrix} 0 & 0 & 0 & 3.8 & 0 & 0 \end{bmatrix}^\top$.

Figure~\ref{fig:car_states} shows the ground truth states and state estimation errors resulting from all three MHE approaches as well as the EKF. 
In Figure~\ref{fig:car_params}, the corresponding estimation errors of the unknown parameters $D_{\mathrm{r}}$ and $D_{\mathrm{f}}$, and the estimates of the tire forces $F_{\mathrm{r}}$ and $F_{\mathrm{f}}$ are shown.
Note that in the lower three subplots of Figure~\ref{fig:car_states} and in Figure~\ref{fig:car_params}, the dotted lines indicate the values at each time step, while the solid lines indicate a moving average thereof over $40$ time steps.
As seen in Figure~\ref{fig:car_states}, the first three states ($x_{\mathrm{p}}$, $y_{\mathrm{p}}$, and $\psi$) are accurately estimated from the noisy measurements by all three MHE approaches and the EKF.
The \textit{MHE-proposed} results in smaller estimation errors for the longitudinal, lateral, and angular velocities, $v_{\mathrm{x}}$, $v_{\mathrm{y}}$, $\omega$, respectively, sometimes over one order of magnitude compared to the two other MHE schemes.
While the EKF results in improved estimates of $v_{\mathrm{x}}$ compared to all three MHE approaches, the estimation errors of $v_{\mathrm{y}}$ and $\omega$ are larger during the right turn compared to the \textit{MHE-proposed}.
As seen in Figure~\ref{fig:car_params}, the \textit{MHE-joint} leads to drifting parameter estimates while the car is driving straight, even with the considered unbiased noise distribution.
When the car is turning, this yields bad estimates of the tire forces and consequently large state estimation errors.
Similarly to the \textit{MHE-joint}, the EKF leads to drifting estimates of $D_{\mathrm{f}}$ while the car is driving straight.
In addition, the estimates of $D_{\mathrm{r}}$ resulting from the EKF have a significantly larger error compared to all three MHE approaches.
The \textit{MHE-no-adaptation} leads to inaccurate estimates of the tire forces during the turn, which also results in increased state estimation errors.
In contrast, the \textit{MHE-proposed} prevents parameter drift when driving straight, while the parameter estimates are improved during the turn and consequently lead to accurate estimates of the tire forces and reduced state estimation errors.
Note that the state estimation errors in $v_{\mathrm{y}}$ and $\omega$ resulting from the \textit{MHE-proposed} increase during the turn compared to when the car is driving straight.
This is due to the biased parameter estimates affecting the dynamics of the car in this phase.
However, they remain below the estimation errors resulting from the \textit{MHE-joint} for joint state and parameter estimation and the \textit{MHE-no-adaptation}.

The median computational time\footnote{The time was taken on a laptop with 12-core Intel i7 processor and 32 GB of memory.} of the \textit{MHE-joint}, \textit{MHE-no-adaptation}, and \textit{MHE-proposed} are $9.3ms$, $7.6ms$, and $8.7ms$, respectively, while the considered time step is $10ms$.
For embedded applications, the computational time can be significantly reduced by
relying on tailored compiled solvers, see, e.g., recent experimental applications of a comparable MHE on a miniature race car~\cite{bodmer2024optimization}.

The considered problem of jointly estimating state and friction parameters in a car is of high practical relevance, since friction parameters are difficult to identify and often varying, e.g., due to aging of tires or changes in the road conditions.
The problem is frequently addressed using MHE, see, e.g.,~\cite{Zanon2013,Boegli2013}.
The presented numerical results show that a standard MHE approach for joint state and parameter estimation can lead to large state and parameter estimation errors.
This can have detrimental effects, especially if the resulting state estimate is used to compute control inputs to a system.
As shown, these issues can be overcome by a very simple design change: The regularization w.r.t. a constant prior~$\bar{\theta}_0$ instead of a parameter estimate $\hat{\theta}_{t-M_t}$ obtained online.
While the proposed regularization w.r.t. a constant prior~$\bar{\theta}_0$ ensures robustness of the resulting state estimates, it still allows for sufficient adaptation and consequently improved state estimates compared to an MHE with no parameter adaptation at all.
\section{Conclusion} \label{sec:sum} 
State estimation algorithms require some form of online adaptation to address nonlinear systems with unknown or time-varying parameters.
In this paper, we showed that standard moving horizon estimation (MHE) approaches for joint state and parameter estimation can fail in the absence of persistency of excitation (PE), e.g., if a system is well regulated or over-parametrized.
To overcome this issue, we proposed the use of a regularization based on a (constant) prior estimate of the unknown parameter within an MHE formulation and established practical stability of the resulting state estimate independent of PE for general nonlinear systems.
We discussed how many existing results from MHE and adaptive estimation relate to the proposed MHE approach.
Finally, we showed practical applicability using a car example.

While the proposed regularization on a given prior parameter estimate allows us to establish stability irrespective of PE, it may deteriorate the estimation performance compared to a standard MHE approach for joint state and parameter estimation in case PE is large.
An interesting question for future research is to extend the proposed MHE approach to achieve an estimation performance comparable to a standard MHE scheme for joint state and parameter estimation in case of large PE, while the robust stability properties irrespective of PE of the proposed approach are preserved.

\section*{References}
\bibliographystyle{IEEEtran}  
\bibliography{parametric_MHE}

\begin{IEEEbiography}
	[{\includegraphics[width=1in,height=1.25in,clip,keepaspectratio]{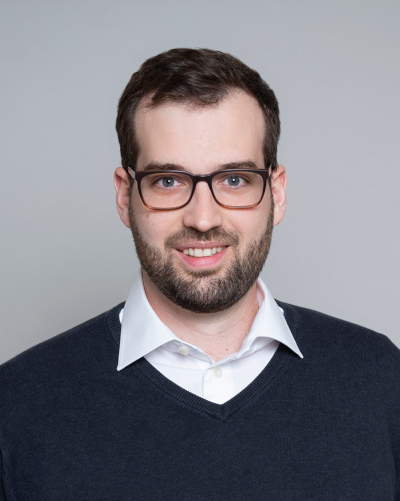}}]{Simon Muntwiler} is currently a sustainability manager in the Office of Sustainability at ETH Zurich, Switzerland. He
	received his Master degree in Robotics, Systems and Control in 2019, and the Ph.D. degree in 2024, both from ETH Zurich, Switzerland.
	During his doctoral studies, he conducted research in the area of optimization- and learning-based state estimation and control algorithms, with application to safety critical systems.
\end{IEEEbiography}

\begin{IEEEbiography}
	[{\includegraphics[width=1in,height=1.25in,clip,keepaspectratio]{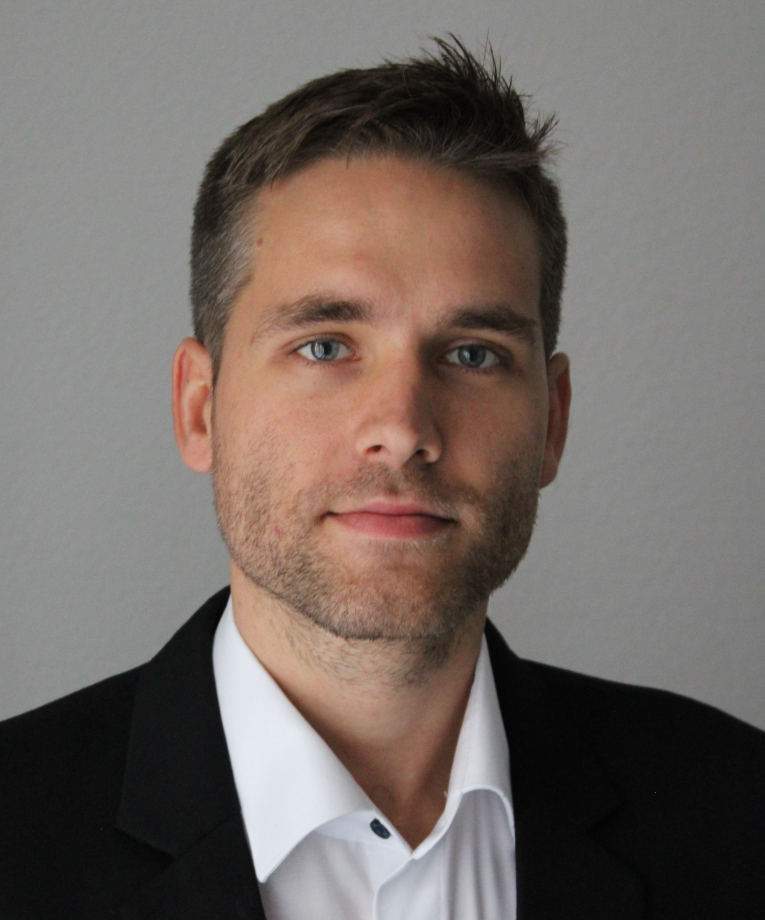}}]{Johannes K\"ohler}
	received his Master degree in Engineering Cybernetics from the University of Stuttgart, Germany, in 2017. In 2021, he obtained a Ph.D. in mechanical engineering, also from the University of Stuttgart, Germany. 
    He is currently a postdoctoral researcher at ETH Zürich, Switzerland. 
    He is the recipient of the 2021 European Systems \& Control PhD Thesis Award, the IEEE CSS George S. Axelby Outstanding Paper Award 2022, and the Journal of process Control Paper Award 2023.
    His research interests are in the area of model predictive control and control and estimation for nonlinear uncertain systems. 
\end{IEEEbiography}

\begin{IEEEbiography}
	[{\includegraphics[width=1in,height=1.25in,clip,keepaspectratio]{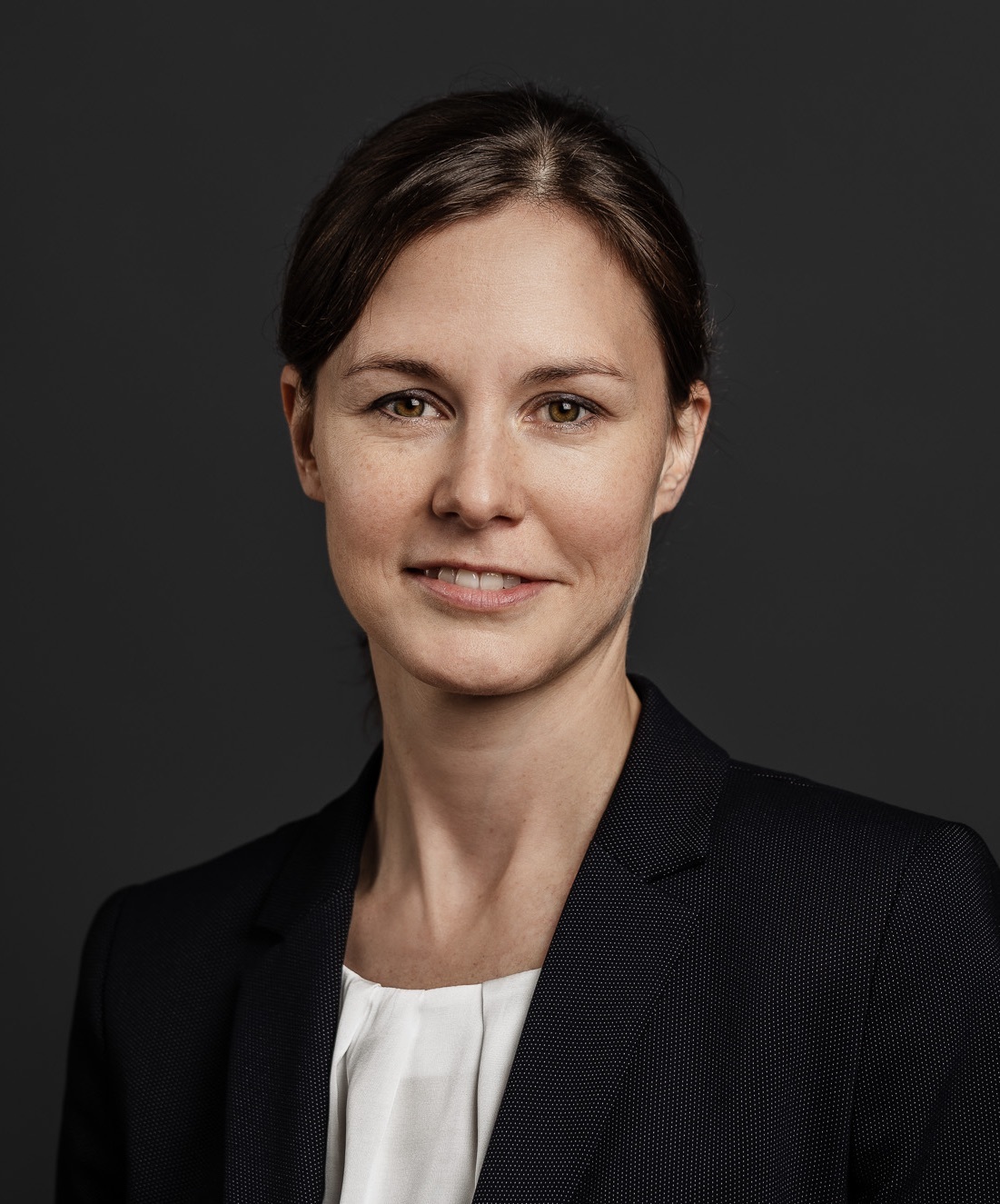}}]{Melanie N. Zeilinger}
 	is an Associate Professor at ETH Zürich, Switzerland. 
	She received the Diploma degree in engineering cybernetics from the University of Stuttgart, Germany, in 2006, and the Ph.D. degree with honors in electrical engineering from ETH Zürich, Switzerland, in 2011. 
	From 2011 to 2012 she was a Postdoctoral Fellow with the Ecole Polytechnique Federale de Lausanne (EPFL), Switzerland.
	She was a Marie Curie Fellow and Postdoctoral Researcher with the Max Planck Institute for Intelligent 	Systems, Tübingen, Germany until 2015 and with the Department of Electrical Engineering and Computer Sciences at the University of California at Berkeley, CA, USA, from 2012 to 2014. 
	From 2018 to 2019 she was a professor at the University of Freiburg, Germany. 
	Her awards include an SNF Professorship grant, the Golden Owl for exceptional teaching at ETH Zurich 2022 and the European Control Award 2023. Her research interests include learning-based control with applications to robotics and biomedical systems.
\end{IEEEbiography}
\appendix
\subsection{Proof of Theorem~\ref{thm:MHE_dIOpS}}\label{app:proof_thm_dIOpS}
This appendix contains the proof of Theorem~\ref{thm:MHE_dIOpS},  which extends ideas from~\cite{Schiller2023} where standard MHE without parametric uncertainty is considered.
We start by noting that a feasible solution to~\eqref{eq:MHE_parametric} is the true state and input sequence, which hence provide an upper bound on the objective~\eqref{eq:MHE_parametric_cost} of the optimal solution:
\begin{align}
	&V_{\mathrm{MHE}}^{\mathrm{x}}(\hat{x}_{t-M_t|t}^*,\hat{w}_{\cdot|t}^*,\hat{y}_{\cdot|t}^*,\hat{\theta}_{t}^*,t) \nonumber \\
	&\quad\le 2\eta^{M_t}\left(\|x_{t-M_t}-\hat{x}_{t-M_t} \|_{P_{2,\mathrm{x}}}^2 + \| \theta - \bar{\theta}_0 \|_{P_{2,\theta}}^2 \right) \label{eq:MHE_proof_objective_bound}  \\
	&\qquad +2\sum_{j=1}^{M_t}\eta^{j-1}\left(\|w_{t-j}-\bar{w}_{t-j}\|_Q^2 +\|y_{t-j}-\bar{y}_{t-j}\|_R^2 \right). \nonumber
\end{align}
Applying Inequality~\eqref{eq:dIOOS_Lyap_2} $M_t$ times we obtain
\begin{align}
W_\delta&(\hat{x}_t,x_t,\hat{\theta}_{t}^*,\theta) \stackrel{\eqref{eq:dIOOS_Lyap_2}}{\le} \eta^{M_t} W_\delta(\hat{x}_{t-M_t|t}^*,x_{t-M_t},\hat{\theta}_{t}^*,\theta) \nonumber \\ 
&+\sum_{j=1}^{M_t}\eta^{j-1}\left(\|\hat{w}^*_{t-j|t}-w_{t-j}\|_Q^2+\|\hat{y}_{t-j|t}^*-y_{t-j}\|_R^2 \right) \nonumber \\
\stackrel{\eqref{eq:dIOOS_Lyap_1}}{\leq}& \eta^{M_t} \left(\|\hat{x}_{t-M_t|t}^*-x_{t-M_t} \|_{P_{2,\mathrm{x}}}^2 + \|\hat{\theta}_{t}^* - \theta \|_{P_{2,\theta}}^2\right) \label{eq:MHE_proof_bound_1} \\
& + \sum_{j=1}^{M_t}\eta^{j-1}\left(\|\hat{w}^*_{t-j|t}-w_{t-j}\|_Q^2 +\|\hat{y}_{t-j|t}^*-y_{t-j}\|_R^2 \right). \nonumber
\end{align}
Applying Cauchy-Schwarz and Young's inequality we obtain
\begin{align*}
\|\hat{x}_{t-M_t|t}^*&-x_{t-M_t} \|_{P_{2,\mathrm{x}}}^2 + \|\hat{\theta}_{t}^* - \theta \|_{P_{2,\theta}}^2 \\
\le& 2 \|\hat{x}_{t-M_t|t}^*-\hat{x}_{t-M_t} \|_{P_{2,\mathrm{x}}}^2 + 2 \|\hat{x}_{t-M_t}-x_{t-M_t} \|_{P_{2,\mathrm{x}}}^2 \\
&+ 2 \| \hat{\theta}_{t}^* - \bar{\theta}_0 \|_{P_{2,\theta}}^2 + 2 \| \bar{\theta}_0 - \theta \|_{P_{2,\theta}}^2.
\end{align*}
Analogously, we have that
\begin{align*}
\|\hat{w}^*_{t-j|t}-w_{t-j}\|_Q^2 \le& 2\|\hat{w}^*_{t-j|t}-\bar{w}_{t-j}\|_Q^2 + 2\|w_{t-j}-\bar{w}_{t-j}\|_Q^2, \\
\|\hat{y}^*_{t-j|t}-y_{t-j}\|_R^2 \le& 2\|\hat{y}^*_{t-j|t}-\bar{y}_{t-j}\|_R^2 + 2\|y_{t-j}-\bar{y}_{t-j}\|_R^2.
\end{align*}
Inserting the above three inequalities in~\eqref{eq:MHE_proof_bound_1}, we arrive at
\begin{align}
W_\delta(\hat{x}_t,&x_t,\hat{\theta}_{t}^*,\theta) \nonumber \\
\stackrel{\eqref{eq:MHE_objective}}{\leq}& V_{\mathrm{MHE}}^{\mathrm{x}}(\hat{x}_{t-M_t|t}^*,\hat{w}_{\cdot|t}^*,\hat{y}_{\cdot|t}^*,\hat{\theta}_{t}^*,t) \nonumber \\
&+  2 \eta^{M_t} \left(\|\hat{x}_{t-M_t}-x_{t-M_t} \|_{P_{2,\mathrm{x}}}^2 + \| \bar{\theta}_0 - \theta \|_{P_{2,\theta}}^2 \right) \nonumber \\
&+ 2\sum_{j=1}^{M_t}\eta^{j-1}\left(\|w_{t-j}-\bar{w}_{t-j}\|_Q^2 +\|y_{t-j}-\bar{y}_{t-j}\|_R^2\right) \nonumber \\
\stackrel{\eqref{eq:MHE_proof_objective_bound}}{\le} & 4 \eta^{M_t} \left(\|\hat{x}_{t-M_t}-x_{t-M_t} \|_{P_{2,\mathrm{x}}}^2 + \| \bar{\theta}_0 - \theta \|_{P_{2,\theta}}^2 \right) \label{eq:MHE_Mt_Lyap} \\
&+ 4\sum_{j=1}^{M_t}\eta^{j-1}\left(\|w_{t-j}-\bar{w}_{t-j}\|_Q^2 +\|y_{t-j}-\bar{y}_{t-j}\|_R^2\right). \nonumber 
\end{align}
We have that
\begin{align*}
\|\hat{x}_{t-M_t}-x_{t-M_t} \|_{P_{2,\mathrm{x}}}^2 \le & \lambda_{\max}(P_{2,\mathrm{x}},P_1)\|\hat{x}_{t-M_t}-x_{t-M_t} \|_{P_1}^2 \\
\stackrel{\eqref{eq:dIOOS_Lyap_1}}{\le} \lambda_{\max}(P_{2,\mathrm{x}}&,P_1)W_\delta(\hat{x}_{t-M_t},x_{t-M_t},\hat{\theta}_{t-M_t}^*,\theta),
\end{align*}
where $\lambda_{\max}(P_{2,\mathrm{x}},P_1) \ge 1$ is the maximum generalized eigenvalue of $P_{2,\mathrm{x}}$ and $P_1$.
Inserting this in~\eqref{eq:MHE_Mt_Lyap} results in
\begin{align}
&W_\delta(\hat{x}_t,x_t,\hat{\theta}_{t}^*,\theta) \le 4 \eta^{M_t} \| \bar{\theta}_0 - \theta \|_{P_{2,\theta}}^2 \nonumber \\
&\quad + 4\sum_{j=1}^{M_t}\eta^{j-1}\left(\|w_{t-j}-\bar{w}_{t-j}\|_Q^2 +\|y_{t-j}-\bar{y}_{t-j}\|_R^2\right) \label{eq:MHE_M_Lyap} \\
&\quad \ +4 \eta^{M_t} \lambda_{\max}(P_{2,\mathrm{x}},P_1)W_\delta(\hat{x}_{t-M_t},x_{t-M_t},\hat{\theta}_{t-M_t}^*,\theta). \nonumber
\end{align}
Consider some time $t = kM + l$ with unique $l\in\mathbb{I}_{[0,M-1]}$ and $k \in \mathbb{I}_{\ge 0}$.
Using Inequality~\eqref{eq:MHE_Mt_Lyap}, we get
\begin{align}
&W_\delta(\hat{x}_l,x_l,\hat{\theta}_{l}^*,\theta) \le 4 \eta^{l} \left(\|\bar{x}_{0}-x_{0} \|_{P_{2,\mathrm{x}}}^2 + \| \bar{\theta}_0 - \theta \|_{P_{2,\theta}}^2\right) \nonumber\\
&\ \ + 4\sum_{j=1}^{l}\eta^{j-1}\left(\|w_{l-j}-\bar{w}_{l-j}\|_Q^2 +\|y_{l-j}-\bar{y}_{l-j}\|_R^2\right). \label{eq:MHE_bound_on_x_l}
\end{align}
Introducing
\begin{align*}
\rho^M \coloneqq 4\eta^{M}\lambda_{\max}(P_{2,\mathrm{x}},P_1) \stackrel{\eqref{eq:M_cond}}{<} 1
\end{align*}
and applying Inequality~\eqref{eq:MHE_M_Lyap} $k$ times using $M_t=M$ we have
\begin{align*}
W_\delta&(\hat{x}_t,x_t,\hat{\theta}_{t}^*,\theta) \le  \rho^{kM} W_\delta(\hat{x}_l,x_l,\hat{\theta}_{l}^*,\theta) \\
& + 4\sum_{i=0}^{k-1}\rho^{iM}\sum_{j=1}^{M}\eta^{j-1}\|w_{t-iM-j}-\bar{w}_{t-iM-j}\|_Q^2 \\
& + 4\sum_{i=0}^{k-1}\rho^{iM}\sum_{j=1}^{M}\eta^{j-1}\|y_{t-iM-j}-\bar{y}_{t-iM-j}\|_R^2 \\
& + \sum_{i=1}^{k}\rho^{iM}\| \bar{\theta}_0 - \theta \|_{P_{2,\theta}}^2 \\
\stackrel{\eqref{eq:MHE_bound_on_x_l}}{\le} & 4 \rho^t \left(\|\bar{x}_{0}-x_{0} \|_{P_{2,\mathrm{x}}}^2 + \| \bar{\theta}_0 - \theta \|_{P_{2,\theta}}^2\right) \\
&+ 4\rho^{kM}\sum_{j=1}^{l}\eta^{j-1}\|w_{l-j}-\bar{w}_{l-j}\|_Q^2 \\
&+ 4\rho^{kM}\sum_{j=1}^{l}\eta^{j-1}\|y_{l-j}-\bar{y}_{l-j}\|_R^2 \\
&+ 4 \sum_{i=0}^{k-1}\sum_{j=1}^{M} \rho^{iM+j-1}\|w_{t-iM-j}-\bar{w}_{t-iM-j}\|_Q^2 \\
&+ 4 \sum_{i=0}^{k-1}\sum_{j=1}^{M} \rho^{iM+j-1}\|y_{t-iM-j}-\bar{y}_{t-iM-j}\|_R^2 \\
&+ \sum_{i=1}^{k}\rho^{iM}\| \bar{\theta}_0 - \theta \|_{P_4}^2 \\
\le & 4 \rho^t \left(\|\bar{x}_{0}-x_{0} \|_{P_{2,\mathrm{x}}}^2 + \| \bar{\theta}_0 - \theta \|_{P_{2,\theta}}^2\right) \\
&+ 4 \sum_{q=0}^{t-1} \rho^{q}\|w_{t-q-1}-\bar{w}_{t-q-1}\|_Q^2 \\
&+ 4 \sum_{q=0}^{t-1} \rho^{q}\|y_{t-q-1}-\bar{y}_{t-q-1}\|_R^2 \\
&+ \frac{\rho^M}{1-\rho^M}\| \bar{\theta}_0 - \theta \|_{P_{2,\theta}}^2,
\end{align*}
where we used that $\eta \le \rho$ and additionally a geometric sum argument in the last step.
For any matrix $P=P^\top \succ 0$ we have that
\begin{align}\label{eq:pd_bounds}
	\lambda_{\min}(P)\|x\|^2 \le \|x\|_P^2 \le \lambda_{\max}(P)\|x\|^2.
\end{align}
Applying the upper bound in~\eqref{eq:pd_bounds} above we obtain
\begin{align*}
W_\delta(\hat{x}_t,x_t,\hat{\theta}_{t}^*,\theta) \le & 4 C_P \rho^t \left\|\begin{bmatrix}
x_0  \\
\theta
\end{bmatrix} - \begin{bmatrix}
\bar{x}_0 \\
\bar{\theta}_0
\end{bmatrix}\right\|^2 \\
&+ 4 \lambda_{\max}(Q)\sum_{q=0}^{t-1} \rho^{q}\|w_{t-q-1}-\bar{w}_{t-q-1}\|^2 \\
&+ 4 \lambda_{\max}(R)\sum_{q=0}^{t-1} \rho^{q}\|y_{t-q-1}-\bar{y}_{t-q-1}\|^2 \\
&+ \lambda_{\max}(P_{2,\theta})\left(\frac{\rho^M}{1-\rho^M}\right)\| \bar{\theta}_0 - \theta \|^2,
\end{align*}
with $C_P = \max\{\lambda_{\max}(P_{2,\mathrm{x}}),\lambda_{\max}(P_{2,\theta})\}$.
Applying the lower bound in~\eqref{eq:dIOOS_Lyap_1} and~\eqref{eq:pd_bounds}, and using the fact that $\sqrt{a+b} \le \sqrt{a} + \sqrt{b}$ for all $a,b \ge 0$, we obtain
\begin{align*}
\|\hat{x}_t-x_t\| \le& \frac{1}{\sqrt{\lambda_{\min}(P_1)}} \sqrt{W_{\delta}(\hat{x}_t,x_t,\hat{\theta}_{t}^*,\theta)} \\
\le& 2 \sqrt{\frac{C_P}{\lambda_{\min}(P_1)}} \sqrt{\rho}^t \left\|\begin{bmatrix}
x_0  \\
\theta
\end{bmatrix} - \begin{bmatrix}
\bar{x}_0 \\
\bar{\theta}_0
\end{bmatrix}\right\| \\
&+ 2 \sqrt{\frac{\lambda_{\max}(Q)}{\lambda_{\min}(P_1)}}\sum_{q=0}^{t-1} \sqrt{\rho}^{q}\|w_{t-q-1}-\bar{w}_{t-q-1}\| \\
&+ 2 \sqrt{\frac{\lambda_{\max}(R)}{\lambda_{\min}(P_1)}}\sum_{q=0}^{t-1} \sqrt{\rho}^{q}\|y_{t-q-1}-\bar{y}_{t-q-1}\| \\
&+ C_\rho\sqrt{\frac{\lambda_{\max}(P_{2,\theta})}{\lambda_{\min}(P_1)}}\sqrt{\rho}^M\| \bar{\theta}_0 - \theta \|,
\end{align*}
where we used that $\nicefrac{1}{(1-\rho^M)} \le \nicefrac{1}{(1-\rho)}$ for $M\in\mathbb{I}_{\ge 1}$ and with $C_\rho = \nicefrac{1}{\sqrt{1-\rho}}$.
This sum-based bound also implies a max bound, compare, e.g.,~\cite[Cor. 1]{Schiller2023},
\begin{align}
\|\hat{x}_t&-x_t\| \nonumber \\
\le& \max\left\{\vphantom{\frac{4}{1-\sqrt[4]{\rho}}}8\sqrt{\frac{C_P}{\lambda_{\min}(P_1)}} \sqrt{\rho}^t \left\|\begin{bmatrix}
x_0  \\
\theta
\end{bmatrix} - \begin{bmatrix}
\bar{x}_0 \\
\bar{\theta}_0
\end{bmatrix}\right\|, \right. \nonumber \\
& \left. \max_{q\in\mathbb{I}_{[0,t-1]}}\left\{\frac{8}{1-\sqrt[4]{\rho}}\sqrt{\frac{\lambda_{\max}(Q)}{\lambda_{\min}(P_1)}}\sqrt[4]{\rho}^{t-j-1}\|w_{j}-\bar{w}_{j}\|\right\}, \right. \nonumber \\
& \left. \max_{q\in\mathbb{I}_{[0,t-1]}}\left\{\frac{8}{1-\sqrt[4]{\rho}}\sqrt{\frac{\lambda_{\max}(R)}{\lambda_{\min}(P_1)}}\sqrt[4]{\rho}^{t-j-1}\|y_{j}-\bar{y}_{j}\|\right\}, \right. \nonumber \\
& \left. \vphantom{\frac{4}{1-\sqrt[4]{\rho}}} 4C_\rho\sqrt{\frac{\lambda_{\max}(P_{2,\theta})}{\lambda_{\min}(P_1)}}\sqrt{\rho}^M\| \bar{\theta}_0 - \theta \| \right\}, \label{eq:MHE_dIOpS_bound} 
\end{align}
where we denoted $j\coloneqq t-q-1$. This shows that~\eqref{eq:ISpS_max_form} is satisfied and consequently, the MHE estimator~\eqref{eq:MHE_parametric} is an exponentially $\delta$-IOpS estimator with respect to $\epsilon$ according to Definition~\ref{def:dIOpS}, which completes the proof. \hfill $\blacksquare$

\end{document}